\documentclass[10pt]{elsart}               % The use of LaTeX2e is preferred.
\usepackage[papersize={8.5in,11in}, margin=1in]{geometry} % letter
\usepackage{tabularx}                 % For tables
\usepackage{bm}                       % Bold math
\usepackage{import}
\usepackage{amsmath,amssymb,amsfonts} % Math stuff
\usepackage{booktabs}                 % Nicer tables
\usepackage{makecell}                 % Better table cells
\usepackage{tikz}                     % Drawing
\usepackage{blkarray}                 % Block arrays
\usepackage{mathtools}                % Math extensions
\usepackage{cuted}                    % Stripped equations
\usepackage{cite}                     % Citations
\usepackage{array}                    % Table columns
\usepackage{stackengine}              % Stacked symbols
\usepackage[linktocpage=true]{hyperref} % Only ONCE, and LAST
\usepackage{cleveref}                 % Clever referencing
\usepackage{cases,empheq}             % Cases in equations
\usepackage{mdframed}                 % Framed environments
\usepackage{xcolor,soul,framed}       % Colors, highlighting
\colorlet{shadecolor}{yellow}
\usepackage[linktocpage=true]{hyperref} % Only ONCE, and LAST
\newcommand{\tr}{\textcolor{black}}
\newcommand{\tb}{\textcolor{black}}
\newcommand{\mt}{\mathtt}

\newcommand{\mb}{\mathbf}

\newcommand{\N}{\mathsf{N}}
\newcommand{\x}{\bm{x}}
\newcommand{\1}{\mathtt{A}}
\newcommand{\2}{\mathtt{B}}
\newcommand{\Q}{\mathtt{Q}}
\newcommand{\s}{\mathtt{s}}

\newcommand{\p}{\mathsf{p}}
\newcommand{\pp}{\mathsf{P}}

\newcommand{\y}{\bm{y}}

\newcommand{\V}{\mathcal{V}}
\newcommand{\Z}{\mathbb{Z}}

\newcommand{\X}{\mathcal{X}}
\def\mathbi#1{\textbf{\em #1}}
\newtheorem{theorem}{Theorem}
\newtheorem{lemma}{Lemma}
\newtheorem{assumption}{Assumption}
\newtheorem{remark}{Remark}
\newtheorem{corr}{Corollary}
\newtheorem{definition}{Definition}
\newtheorem{proposition}{Proposition}
\newtheorem{example}{Example}
\newtheorem{revisit}{Revisiting Example}

\newtheorem{lemmaAppendix}{Lemma}
\newenvironment{proof}[1][Proof]{\par\noindent\textit{#1.} }{\hfill$\square$\par}
\newenvironment{proofref}[1]
{\par\noindent\textit{Proof of {#1}.} }
{\hfill$\square$\par}

\begin{document}
	
	\begin{frontmatter}
		
		\title{From Discrete to Continuous Mixed Populations of Conformists, Nonconformists, and Imitators} % Title, preferably not more 

		\author[Rome]{Azadeh Aghaeeyan}    % Add the 
		\author[Cal]{Pouria Ramazi}          % e-mail address 
		% (ead) as shown
				\address[Rome]{Department of Mathematics and Statistics, Brock University, Canada.}  
		\address[Cal]{Department of Biological Sciences and Department of Mathematics and Statistics, University of Calgary, Canada.}  % Please supply                                              
           % full addresses
		% here.

		% \begin{keyword}                           % Five to ten keywords,  
			% Game Theory, Population Dynamics, Convergence Analysis               % chosen from the IFAC 
			% \end{keyword}                             % keyword list or with the 
		% help of the Automatica 
		% keyword wizard

		\begin{abstract}                          % Abstract of not more than 200 words.
			In two-strategy decision-making problems, individuals often imitate the highest earners or choose either the common or rare strategy.
			Individuals who benefit from the common strategy are \emph{conformists}, whereas those who profit by choosing the less common one are called \emph{nonconformists}.
			The population proportions of the two strategies may undergo perpetual fluctuations
			in finite, discrete, heterogeneous populations of imitators, conformists, and nonconformists.
			How these fluctuations evolve as population size increases was left as an open question and is addressed in this paper.
			We show that the family of Markov chains describing the discrete population dynamics forms a generalized stochastic approximation process for a differential inclusion—the continuous-time dynamics.
			Furthermore, we prove that the continuous-time dynamics always equilibrate. 
			Then, by leveraging results from the stochastic approximation theory, we show that the amplitudes of fluctuations in the proportions of the two strategies in the population approach zero with probability one when the population size grows to infinity.
			Our results suggest that large-scale perpetual fluctuations are unlikely in large, well-mixed populations consisting of these three types, particularly when imitators follow the highest earners.
		\end{abstract}
		
	\end{frontmatter}
	
	\section{Introduction}
	Many everyday decisions, such as whether to follow a fashion trend, join a volunteer group, or get a flu shot, involve choosing between two competing strategies \cite{zhang2018fashion, kawagoe2023asymmetric, chapman2012using}.
	\emph{Imitation} and \emph{best response} are two main rules commonly used to explain individual decision making, either in two-strategy scenarios \cite{bauch2004vaccination, 10172285,10416815} or in situations involving multiple strategies \cite{comoImitation, replicator, replicator2, barreiro2018constrained}.
	Under \tb{one of the flavors of} imitation update rules, individuals copy the strategies of the highest-earning peers, whereas according to the best-response update rule, individuals choose the strategy that maximizes their payoff.
	In particular, \emph{conformist} best-responders benefit from strategies already adopted by a portion of the population greater than their \emph{threshold}. 
	\emph{Nonconformist} best-responders, on the other hand, favor strategies that are uncommon relative to their thresholds \cite{kaniovski2000adaptive,coordinationandanticoordination,arditt2,haslegrave2017majority}.
	These behaviors are captured by \emph{linear threshold models} and \emph{anti-threshold models}, respectively.
	
	Much effort has been devoted to characterizing the asymptotic behavior of populations composed solely of best responders \cite{7258336,10453658,ZHU2023110707,bestresponsePotential,arefizadeh2023robustness} or only imitators \cite{comoImitation,sandholm2010population}, in both structured and well-mixed settings. 
	Additionally, several studies have proposed control mechanisms to steer the population toward a desired distribution of strategies \cite{zino2025equilibrium,7458850,7438812}.
	
	In contrast, to the best of our knowledge, very few studies have investigated populations that include both best-responders and imitators \tb{with almost arbitrarily payoff matrices}.
	In \cite{Hien2}, the dynamics of finite discrete heterogeneous populations of imitators, conformists, and nonconformists were formulated, and their equilibria were characterized. 
	It has been shown that, in the long run, the population may settle into an  
	equilibrium, where each individual is satisfied with her strategy, or undergo perpetual fluctuations, where some individuals keep revisiting their strategies, resulting in fluctuations in the population proportion of one of the two strategy players.
	
	% Perpetual fluctuations as asymptotic behavior have also been reported for populations consisting entirely of imitators \cite{fu2024evolutionary} or entirely of conformists and nonconformists \cite{roohi}.
	% In \cite{aghaeeyan2023discrete, aghaeeyanCDC}, using results from stochastic approximation theory \cite{roth2013stochastic,benaim2005stochastic,benaim1998recursive}, it has been shown that the amplitudes of the fluctuations in these populations converge almost surely to zero as the population size approaches infinity.
	% Whether a similar conclusion holds for mixed populations of imitators, conformists, and nonconformists remains unsolved.
	
	\tb{But what would be the asymptotic behavior of such a population when the population size approaches infinity?
		To answer this question, we}
	adopt an approach similar to \cite{aghaeeyan2023discrete, aghaeeyanCDC}:
	We \emph{(i)} show that the discrete population dynamics form a Markov chain; \emph{(ii)}  show that the Markov chain, when indexed by population size, builds a \emph{generalized stochastic approximation process} for a differential inclusion; and \emph{(iii)}  analyze the corresponding differential inclusion to infer convergence properties.
	To achieve this, we expand the state space to include the dynamics of imitative subpopulations. 
	Additionally, to effectively model the imitators' dynamics, inspired by \cite{fu2024evolutionary}, we introduce the notions of maximum utilities and maximum active utilities.
	The coexistence of imitators, conformists, and anticonformists leads to new equilibria, specifically, continua of equilibrium points.
	
	In \cite{aghaeeyan2023discrete}, we studied the differential inclusion through the analysis of a scalar differential inclusion--\emph{the abstract dynamics}. 
	The abstract dynamics were equivalent to the dynamics of the population proportion of one of the two strategies (Strategy $\1$).
	Here, we apply a similar analytical framework, though significant differences arise due to the inclusion of imitators, so the abstract dynamics no longer align directly with the proportion dynamics of $\1$-players.
	Consequently, trajectories starting at boundary points require separate analysis. 
	A related scenario involving imitators with diagonal anticoordination payoff matrices was studied in \cite{aghaeeyanCDC}, which handled boundary points through restrictive conditions on population distributions. In contrast, our current model is more flexible, allowing payoff matrices and population distributions to be almost arbitrary.

	\tb{The long-run behavior of populations—whether trajectories converge to equilibria or exhibit persistent fluctuations—has been studied extensively in evolutionary game theory and learning in games. These studies mostly use homogeneous revision protocols, such as best response, imitation, and their variants, and examine their deterministic limits; see, e.g., \cite{sandholm2010population,arefizadeh2023robustness,9234634}. In parallel, conformist (threshold) and nonconformist/contrarian (anti-threshold) behaviors are central in social-influence models \cite{granovetter1978threshold,galam2004contrarian}. Most available asymptotic characterizations, however, focus on populations in which all agents use the same decision rule (e.g., \cite{ZHU2023110707,bestresponsePotential}), or treat only limited heterogeneous mixtures, such as conformists and nonconformists without imitation (e.g., \cite{grabisch2019model,arditt2}). In this context, our contribution is threefold. (i) We derive the differential inclusion corresponding to the discrete dynamics of a mixed population of imitators, conformists, and nonconformists. We characterize its equilibria, including new mixed/non-mixed nonconformist-driven and mixed/non-mixed conformist-driven equilibrium classes (Lemma \ref{lem_equilibriumPointofContinuous}). These classes were absent in \cite{Hien2} due to the exclusion of integer thresholds and, in the mixed case, form continua rather than isolated points. (ii) We analyze the asymptotic behavior of the differential inclusion and show that solution trajectories always converge to these equilibria (\Cref{lem:attraction-population-single-equilibrium}).  (iii) We show that, as the population size approaches infinity, the amplitude of perpetual fluctuations in the proportion of $\1$-players converges to zero almost surely (\Cref{thm_birkhoff_center_invariant_population} and \Cref{corollary}).}
	\subsection*{Notations.} 
	Denote the set of real numbers and non-negative integers by $\mathbb{R}$ and $ \Z_{\geq 0}$.
	\tb{Let} $\Z^{n}$ and $\mathbb{R}^n$ denote the sets of $n \times 1$ vectors with integer and real elements, respectively.
	For a positive integer $\N$, the notation $\frac{1}{\N}\Z^{n}$ denotes the set of  vectors of size $n$ with rational components with denominators dividing $\N.$
	A calligraphic font is used to denote sets, and the boundary of a set $\X$ is denoted by $\partial \X.$
	Vectors are denoted in boldface, and $x_p$ represents the $p^{\text{th}}$ element of vector $\x.$
	The notation $[k]$ for a positive integer $k$ denotes $\{1,2,\ldots,k \}$. % and for $k=0$ denotes the singleton set $\{0\}$.
	For a real number (resp. logical variable) $x$, $\mathtt{1}(x)$ equals one if $x>0$ (resp. $x$ is $\mathtt{TRUE}$) and $-\infty$ otherwise.
	A sequence of random variables $X_0, X_1, \ldots,$ is denoted by $\langle X_k\rangle_{k=0}^\infty.$
	A correspondence is denoted by $\rightrightarrows.$
	Sign function, denoted by $\operatorname{sgn}(x),$ returns one for a positive argument, negative one for a negative argument and zero otherwise. 
	For two strategies $\mathtt{x}$ and $\mathtt{y}$, $\mathtt{x}, \mathtt{y} \in \{\1, \2\}$, the statement $\mathtt{x}==\mathtt{y}$ is $\mathtt{TRUE}$ if  $\mathtt{x}$ and  $\mathtt{y}$ have the same value, and $\mathtt{FALSE}$ otherwise.
	%(Commented to save sapce)Let $m \in \mathbb{Z}_{\geq 0}$, and for each $j = 1,2,\ldots,m$, let $n_j \in \mathbb{Z}_{\geq 0}$ and $[a_i^j,b_i^j] \subset \mathbb{R}$ for $i = 1,2,\ldots, n_j$. 
	%Then, the generalized Cartesian product is $\prod_{j=1}^m (\prod_{i=1}^{n_j} [a^j_i,b^j_i])$ which is a subset of $\mathbb{R}^{n}$ with $n = \sum_{j=1}^m n_j.$
	Table \ref{table} summarizes the most frequently used notations throughout the paper.
	\begin{table}
		\caption{\textbf{TABLE OF NOTATION}}\label{table}
		\renewcommand{\arraystretch}{1.2}
		\begin{tabularx}{\linewidth}{@{}lX@{}}
			\toprule
			$\1$, $\2$ & possible strategies in the game\\
			$\mathtt{s}^i$ & strategy of agent $i$  \\
			$a^i$ (resp. $c^i$) & payoff of agent $i$ for playing strategy $\1$ (resp. $\2$) against an $\1$-player \\
			$b^i$ (resp. $d^i$) & payoff of agent $i$ for playing strategy $\1$ (resp. $\2$) against a $\2$-player \\
			%$c^i$ & payoff of agent $i$ for playing strategy $\2$ against an $\1$-player \\
			%$d^i$ & payoff of agent $i$ for playing strategy $\2$ against a $\2$-player \\
			$\mathtt{u}^{\1,i}(x)$ (resp. $\mathtt{u}^{\2,i}(x) $)&  utility of agent $i$ for playing strategy $\1$ (resp. $\2$) \\
			%$\mathtt{u}^{\2,i}(x^\N) $&  utility of agent $i$ for playing strategy $\2$ in a population of size $\N$ \\
			$\tau^i$&  threshold of agent $i$ \\
			%$\tau_{i}$&  threshold of type $i$ \\
			$\tau_{i}'$ (resp. $\tau_{i}$) &  threshold of conformist (resp. nonconformist) type $i$  \\
			$\p'$ (resp. $\p$) & total number of (conformist) nonconformist types \\
			%$\p'$ & total number of conformist types \\
			$\pp$ & $\p + \p'$ \\
			%$\rho_i$ &  population proportion of nonconformist type $i$ \\
			$\rho_i'$ (resp. $\rho_i$) &  population proportion of conformist (nonconformist) type $i$\\
			%$\zeta_i$ &  population proportion of imitators sharing the payoff matrix with nonconformist type $i$ \\
			$\zeta_i'$ (resp. $\zeta_i$) &  population proportion of imitators sharing the payoff matrix with conformist (resp. nonconformist) type $i$ \\
			$\x$ &  population state\\
			$u^{\1}_i(x)$ (resp. $u^{\2}_i(x) $)&  utility of type $i$ for playing strategy $\1$ (resp. $\2$)  \\
			%$u^{\2}_i(x) $&  utility of type $i$ for playing strategy $\2$ \\
			$\mb{X}^{\frac{1}{\N}}$ &  Markov chain state of population of size $\N$\\
			$\mathcal{T}$ & set of thresholds \\
			$\mathcal{C}$ & set of intersection points of  utilities\\
			$\bm{\mathcal{X}}_{ss}$ & state space of the continuous-time population dynamics\\
			\bottomrule
		\end{tabularx}
	\end{table}
	
	% \begin{figure}
		% \begin{center}
			% The printed column  
			% \caption{Gaius Julius Caesar, 100--44 B.C.}  % width is 8.4 cm.
			% \label{fig1}                                 % Size the figures 
			% \end{center}                                 % accordingly.
		% \end{figure}
	
	% OR
	
	%\begin{figure}
	%\begin{center}
	%\epsfig{file=jcaesar,width=7cm}
	%\caption{Gaius Julius Caesar, 100--44 B.C.}
	%\label{fig1}
	%\end{center}
	%\end{figure}
	
	\section{Problem Formulation} \label{sec_problemFormulation}
	\tb{
		We study the long-term behavior of a well-mixed population in which agents are activated one at a time; upon activation, an agent plays a two-strategy game against the entire population and receives the corresponding  payoff.
		The game is a \emph{matrix game}, meaning payoffs for all possible pairwise outcomes  can be summarized as a matrix.
		Agents are heterogeneous in their payoff matrices and are either a best responder or an imitator.
		We provide the formal model next.
	}
	
	Consider a well-mixed population of $\N$ agents, labeled by $1,2,\ldots,\mathsf{N}$, playing two-strategy games over time \tb{$t \in \frac{1}{\mathsf{N}}\mathbb{Z}_{\geq 0}$, indexed by $k$ where $k =\N t$,} and, accordingly, obtaining payoffs.
	A $2\times2$ matrix summarizes the four possible payoff gains to agent $i$ against another agent:
	\vspace{-10pt}
	\begin{equation*}
		\tb{\pi^i }=
		\raisebox{0.3ex}{$
			\begin{blockarray}{ccc}
				& \1 & \2\\
				\begin{block}{c(cc)}
					\1 & a^{i} & b^{i}\\
					\2 & c^i & d^i\\
				\end{block}
			\end{blockarray}
			$}
	\end{equation*}
	\vspace{-25pt}
	
	where $a^i,b^i,c^i,d^i \in \mathbb{R}$ and $a^i$ and $b^i$ (resp. $c^i$, and $d^i$) are the payoffs that agent $i$ obtains by playing strategy $\1$ (resp. $\2$) against an $\1$-player and a $\2$-player, respectively.   
	In a well-mixed population, each agent plays against all other agents, including herself, and obtains an accumulated payoff.
	The average payoff or \emph{utility} of agent $i$, when the population proportion of strategy-$\1$ players equals $x$,
	is
	\begin{align*} 
		\mathtt{u}^{\1,i}\tb{(k)} &= \tb{\begin{pmatrix} 1, & 0\end{pmatrix} \pi^i \begin{pmatrix}x(k),& 1 - x(k) \end{pmatrix}^\top}\\
		&=
		(a^i - b^i)x\tb{(k)} + b^i 
	\end{align*}
	if she plays strategy $\1$, and
	\begin{align*}
		\mathtt{u}^{\2,i}\tb{(k)} &= \tb{\begin{pmatrix} 0, & 1\end{pmatrix} \pi^i \begin{pmatrix}x(k),& 1 - x(k) \end{pmatrix}^\top}\\
		&= (c^i - d^i)x\tb{(k)} + d^i
	\end{align*}
	if she plays strategy $\2$.
	At each time index $k$, only one agent becomes active and can revise her strategy.
	
	Each agent is either an \emph{imitator}, a \emph{conformist} (also known as a \emph{coordinator}), or a \emph{nonconformist} (also known as an \emph{anticoordinator}).
	An imitator updates her strategy according to the \emph{imitation update rule}, i.e., she switches to the strategy played by the highest earners. 
	More specifically, at time index $k+1$, the strategy of the imitative agent $i$ active at time index $k$ will be
	\begin{equation}
		\scalebox{0.95}{$
			\begin{aligned} \label{eq:imitation}
				&\s^i(k+1)  
				\! =\! \begin{cases}\!
					\! \1, 
					\text{ if } \displaystyle \max_{j\in \N} \mathtt{u}^{\1,j}\tb{(k)}\mt{1}\big(\s^j(k) \mathrel{\!==\!}\1\big)
					\\
					\quad \geq \displaystyle \max_{j\in \N} \mathtt{u}^{\2,j}\tb{(k)}\mt{1}\big(\s^j(k) \!\mathrel{==\!}\2\big),   \\
					\! \2,   \text{ otherwise},
				\end{cases}
			\end{aligned}$}
	\end{equation}
	where $\s^j(k)$ is the strategy of agent $j$ at time index $k.$
	The  imitation update rule \eqref{eq:imitation} dictates the same strategy for any active imitator, regardless of her payoff matrix.
	\tb{
    According to the update rule (1), an imitator adopts the strategy currently used by the highest earners. We interpret this as a parsimonious model of success-/payoff-biased social learning: when agents cannot reliably evaluate their own payoff function, they may use observable success cues to decide whom to copy. 
    Such ``copy-the-successful'' behavior has been documented in humans \cite{henrich2001evolution},
     and other animals \cite{BarrettEtAl2017,BonoWhitenEtAl2018}
    and ``imitate-the-best'' is a standard update rule in evolutionary and spatial game models  \cite{nowak1992evolutionary}.
	}
	
	Upon a revision opportunity, both conformists and nonconformists switch to the strategy that maximizes their utilities based on the current status.
	That is,
	the active agent $i$, who is either conformist or nonconformist, at time index $k$, goes for strategy $\1$ if
	$\mathtt{u}^{\1,i}(x)$ is not less than 
	$\mathtt{u}^{\2,i}(x)$ and chooses strategy $\2$ otherwise.
	If agent $i$ is a conformist, then $a^i + d^i - b^i - c^i >0$, and her update rule becomes
	\begin{align}
		&\s^i(k+1)  
		= \begin{cases}
			\1,& \text{if } x(k) \geq {\tau^i},   \\
			\2,
			& \text{otherwise},
		\end{cases}
		\label{eq:scor}
	\end{align}
	where $\tau^i \triangleq (d^i-b^i)/(a^i + d^i - b^i - c^i)$ is the \emph{threshold} of agent $i$.
	Thus, a conformist adopts the commonly played strategy, based on her perceived prevalence, which is determined by her threshold.
	If $a^i + d^i - b^i - c^i <0,$ agent $i$ is a nonconformist playing the uncommon strategy according to the following update rule:
	\begin{align}
		&\s^i(k+1)  
		= \begin{cases}
			\1,& \text{if } x(k) \leq \tau^i,   \\
			\2,
			& \text{otherwise.}
		\end{cases}
		\label{eq:santi}
	\end{align}
	If $a^i + d^i - b^i - c^i =0,$ agent $i$ is always an $\1$-player (resp. a $\2$-player) whenever $b^i \geq d^i$ (resp.  $b^i < d^i$).
	Conformists (resp. nonconformists) sharing the same payoff matrix and, in turn, threshold form a \emph{type}. 
	In total, we assume there are $\p'$ (resp. $\p$) distinct and non-empty types of conformists (resp. nonconformists). 
	Conformist types are labeled by $1,\ldots,\p'$ with increasing thresholds as $\tau'_1 < \tau'_2 < \dots < \tau'_{\p'}$, whereas nonconformist types are labeled by $1,\ldots,\p$ with decreasing thresholds as $\tau_1 > \tau_2 > \dots > \tau_{\p}$. 
	Here, $\tau'_p$ (resp. $\tau_p$) denotes the threshold associated with the $p^{\text{th}}$ conformist (resp. nonconformist) type.
	We assume that the thresholds of the types are unique and strictly between zero and one and denote the set of all thresholds by $\mathcal{T} =  \{\tau_1,\ldots, \tau_{\p}, \tau'_1,\ldots, \tau'_{\p'} \}$.
	We assume that for each imitator, there exists exactly one conformist or nonconformist type that shares the same payoff matrix.
	The heterogeneity of the population is hence captured by the distribution of population proportions over the total $2\pp$ types, with $\pp = (\p + \p'),$ as
	$$\bm{\theta} = (\underbrace{{\zeta}_1, \ldots, {\zeta}_{\p}, {\zeta}'_{\p'}, \ldots, {\zeta}'_1}_{\pp\text{ imitative types}}, \underbrace{{\rho}_1, \ldots, {\rho}_{\p}}_{\substack{\p\text{ nonconformist}\\\text{types}}}, \underbrace{{\rho}'_{\p'}, \ldots, {\rho}'_1}_{\substack{\p'\text{ conformist}\\\text{types}}})^\top,$$
	where  
	${\zeta}'_j$ (resp. ${\zeta}_j$) represents the ratio of the number of imitators sharing the payoff matrix with type $j$ conformists (resp. nonconformists) to the population size $\N$.
	Also, ${\rho}'_j$ (resp. ${\rho}_j$) represents the ratio of the number of agents in conformist (resp. nonconformist) type $j$  to the population size $\N$.
	For the convenience of vector indexing, for $\p < j \leq \pp$, we define 
	${\zeta}_j = {\zeta}'_{\pp +1-j}$ and $ 
	{\rho}_j = {\rho}'_{\pp +1-j}$. 
	%and  ${\tau}_j = {\tau}'_{\pp +1-j}$. 
	For example, $\theta_{2\pp+1-p}=\rho_{\pp+1-p}=\rho'_p$ is the population proportion of  conformist type $p$.
	\tb{The utility of an agent in type $p$, $p \in \pp$, for playing $\1$ (resp. $\2$) is denoted by $u^\1_i(x)$ (resp. $u^\2_i(x)$).}
	At the population level, the focus is on collective behavior within each type rather than on individual agents.
	We then define the \emph{population state} $\x$ as 
	a $(2\pp)$-dimensional vector where the first $\pp$ elements denote the distribution of $\1$-players among imitative types, followed by $\p$ elements for nonconformist types, and finally $\p'$ elements for conformist types.
	Equivalently,
	\begin{equation*}\scalebox{0.9}{$
			\x
			\! =\! (\underbrace{{x}_1,
				\ldots, 
				{x}_{{\p}},
				{x}_{\p+1},
				\ldots,
				{x}_{{\pp}}}_{\pp\text{ imitative types}},
			\underbrace{{x}_{\pp+1},
				\ldots, 
				{x}_{{\pp + \p}}}_{\substack{\p\text{ nonconformist}\\\text{types}}},
			\underbrace{{x}_{\pp+\p+1},
				\ldots,
				{x}_{{2\pp}}}_{\substack{\p'\text{ conformist}\\\text{types}}})^\top,$}
	\end{equation*}
    \tb{where} the resulting state space then equals 
	$\bm{\mathcal{X}}_{ss} \cap \frac{1}{\mathsf{N}}\mathbb{Z}^{{2\pp}}$ with
	% The $p^\text{th}$ element of vector $\x^\N$ is denoted by $x^\N_p$.
	$$\bm{\mathcal{X}}_{ss}\!=\! \prod_{j=1}^{{\p}}[0,{\zeta}_j]\! \times\! \prod_{j=1}^{{\p'}}[0,\zeta'_{\p'-j+1}] \times \prod_{j=1}^{{\p}}[0,{\rho}_j]\! \times\! \prod_{j=1}^{{\p'}}[0,\rho'_{\p'-j+1}],$$
	\tb{		\begin{definition} \label{def_max_utilities}
			Define the \emph{maximum $\1$-utility} and \emph{maximum $\2$-utility} respectively by 
			\begin{subequations} \label{eq:UmaxABscalar} \begin{align} u^\1(x) &= \max_{i\in[\pp]} u^\1_i(x),\\ u^\2(x) &= \max_{i\in[\pp]} u^\2_i(x), \end{align}
				\label{eq:unconditional-max-utilities}
			\end{subequations}
			for $x\in[0,1]$.
			Correspondingly, the \emph{maximum active $\1$-utility}, 
			$u^\1(\x),$ and the \emph{maximum active $\2$-utility}, 
			$u^\2(\x),$ are defined by:
			\begin{subequations} \label{eq:UmaxAB} \begin{align} u^\1(\x) &= \max_{i\in[\pp]} \bigl(u^\1_i(x)\mt{1}(x_i + x_{\pp+i})\bigr), \label{eq:UmaxAB-A}\\ u^\2(\x) &= \max_{i\in[\pp]} \bigl(u^\2_i(x)\mt{1}(\theta_i + \theta_{\pp+i} - x_i -x_{\pp+i})\bigr), \label{eq:UmaxAB-B}
				\end{align} 
			\end{subequations}
			for $\x\in\bm{\mathcal{X}}_{ss}$ and where $x = \bm 1^\top \x = \sum_{p=1}^{2\pp}x_p.$
		\end{definition}
	}
	\tb{
		In \Cref{def_max_utilities}, we abuse notation by using the same symbol $u^\1$ (resp. $u^\2$) to denote both the maximum $\1$-utility and the maximum active $\1$-utility (resp. $\2$-utility and maximum active $\2$-utility).
		The former is defined over the unit interval, while the latter is defined over $\bm{\mathcal{X}}_{ss}$.
		The maximum $\1$-utility (resp. $\2$-utility) and the maximum active $\1$-utility (resp. $\2$-utility) are equal if the maximum $\1$-utility (resp. $\2$-utility) is obtained by some agents.
		More specifically, whenever the maximum  $\1$-utility (resp. $\2$-utility)  equals $u^\1_i(x)$ (resp. $u^\2_i(x)$) for some $i \in [\pp]$ and
		state $\x$ satisfies $0<x_{i} + x_{\pp+i}$ (resp. $x_i + x_{\pp+i}< \theta_i + \theta_{\pp+i}$), then the maximum $\1$-utility and maximum active $\1$-utility (resp. maximum $\2$-utility and maximum active $\2$-utility) are equal.
		% If the maximum occurs exactly at, $0$ or $\theta_{\pp+i}$, the maximum active utility at $\x$ will differ from the maximum utility at $\sum_{p=1}^{2\pp}x_p(t)$.
		\tb{Also note the difference between the utility of agent $i$, $i \in [\N],$ denoted by $\mathtt{u}^{\1,i}$ (resp.  $\mathtt{u}^{\2,i}$) and the utility of type-$i$ agent, $i \in [\pp]$,  denoted by ${u}_i^{\1}$ (resp.  ${u}_i^{\2}$).}
        }

        \tb{
		We denote the set of points of intersection of the maximum $\1$-utility and maximum $\2$-utility by $\mathcal{C}$, i.e.,
		\begin{equation*}
			\mathcal{C} = \{x \in [0,1] \mid 
			u^\1(x) = u^\2(x) \}.
		\end{equation*}
	}
	\tb{
    In the following, we define additional notation required to present the results.
    The notation $\eta'_j$ represents the population proportions of conformists whose thresholds are at most $\tau'_j$, whereas $\eta_i$ corresponds to the population proportions of nonconformists whose thresholds are not less than $\tau_i$, i.e., 
		\begin{equation*}
			\eta'_j \!=\! \sum_{p=1}^{j} \rho'_p, \quad
			\eta_i = \sum_{p=1}^{i} \rho_p, 
		\end{equation*}
		where $j \in [\p']$, $i \in [\p]$,
		and we define $\eta'_0 = \eta_0 =  \eta_{-1} = \eta'_{-1} =  0$,
		$\eta'_{\p'+1} = \eta'_{\p'},\eta_{\p+1} = \eta_{\p}$,
		$\tau_{\p+1} = \tau'_0 = 0$, and $\tau_0 = \tau'_{\p'+1} = 1$.
		We additionally denote the proportion of the population that are imitators by
		$\zeta_0$, i.e.,
		$\zeta_0 = \sum_{p=1}^{\pp} \theta_p$.
	}
	\tb{
		\begin{assumption} \label{ass:ass1}
			The type proportions satisfy
			\begin{align*}
				&\forall k \in [\p]\cup\{0\}\forall l \in [\p']\cup\{0\} 
				\\
				&\qquad \qquad  \big(\{\eta_k + \eta'_l, \eta_k + \eta'_l + \zeta_0\} \cap (\mathcal{T} \cup \mathcal{C}) = \varnothing \big).
			\end{align*}
		\end{assumption}
	}
	\tb{		Given \Cref{ass:ass1}, the sums of any combinations of the population proportions types do not equal any of the thresholds or the intersection points of the maximum $\1$-utility and $\2$-utility.
	}
	\tb{
		\begin{assumption} \label{ass:ass2}
			\begin{align*}
				\forall x\in \mathcal{C}  \exists p \in [\pp](x= \tau_p)  \Rightarrow u^\1(x) = u^\1_p(x), u^\2(x) = u^\2_p(x).
			\end{align*}
		\end{assumption}
	}
	\tb{
		According \Cref{ass:ass2}, if the point $x=\sum_{p=1}^{2\pp}x_p$ is a point of intersection of the maximum $\1$-utility and $\2$-utility and equals a threshold $\tau_p$ for some type $p\in\mathcal{T}$, i.e., $x=\tau_p$, then these maximum utilities are the utilities of (an agent of) non-imitative type $p$.
		In other words, if both $\1$ and $\2$ result in the highest utility when $x$ equals the threshold of a type, then no other type earns more than that type at $x$. 
		% If both A and B are resulting in the highest utility when x equals the threshold of an agent, then no other agent earns more than that agent at x. 
		% Refer to \Cref{remak_contradictionOfAssumption2} for the implications of this assumption.
		We further assume that no $\1$-utility lines, i.e., $u^\1_i(x)$,  $i \in [\pp]$, coincide with any $\2$-utility lines, i.e., $u^\2_i(x)$,  $i \in [\pp]$:
		\begin{assumption} \label{ass:ass3}
			\begin{equation*}
				\{(i,j) \in [\pp]\times [\pp] \mid b_i = d_j, a_i-b_i+d_j - c_j =0 \} = \varnothing.
			\end{equation*}
	\end{assumption}}
	\vspace{-10pt}
	The evolution of the state over time characterizes the population dynamics, which we call the \textbf{\emph{discrete mixed population dynamics}}.
	The dynamics are determined by the \emph{activation sequence} of the agents and update rules \eqref{eq:imitation}, \eqref{eq:scor}, \eqref{eq:santi}.
	The activation sequence is a sequence of \tb{mutually independent} random variables $\langle {A}_k \rangle_{k=0}^{\infty}$, where ${A}_k$ denotes the active agent at time index $k$ and follows a uniform distribution over $[\N]$, i.e., 
	$\mathbb{P}[{A}_k = i] = \frac{1}{\N}$.
	
	%Azadeh: I deleted the definition of fluctuation set and persistent activation sequence, following our meeting.
	
	% Checked Pouria: Do the dynamics reach an equilibrium, where...? Otherwise, it means that the dynamics never converge to an equilibrium and then if we focus on the number of A-players over time, we observe that it never settles and exhbitis neven ending fluctuations. We know that the second case can happen. But the question is whether the amplitue of these fluctuations will vanish as the population size grows. 
	
	What is the asymptotic behavior of the population dynamics?
	Do the dynamics reach an equilibrium, where the population proportion of $\1$-players of each type remains fixed?
	Otherwise, it means that the dynamics never converge to an equilibrium, and then if we focus on the number of $\1$-players over time, we observe that it never settles and exhibits persistent fluctuations. 
	We know that the second case can happen. 
	But the question is whether the amplitudes of these fluctuations will vanish as the population size grows. 
	
	\begin{example} \label{example}
\tb{Consider a population of $\N$ individuals in a community who are repeatedly deciding  whether to
(i) remain \emph{enrolled/scheduled} for a free seasonal influenza shot (strategy $\1$) or
(ii) remain \emph{unenrolled} (strategy $\2$).
Enrollment is reversible. One can schedule, cancel, and reschedule during the registration period. 
Thus, switching between $\1$ and $\2$ over time is meaningful.
Individuals receive  \emph{revision opportunities}
via an appointment reminder.
We model these revision opportunities with a uniform random activation sequence.
Individuals are of three behavioral categories:
\emph{conformists}, \emph{nonconformists}, and \emph{imitators}.
Upon activation, conformists/nonconformists apply the update rules \eqref{eq:scor}-\eqref{eq:santi}, which capture compliance with social norms and free-riding tendencies, respectively.
Imitators, however, apply \eqref{eq:imitation}: they adopt the strategy currently used by the \emph{highest earners}.
We consider $\pp=4$ payoff-matrix types  with population proportions}
		\begin{equation*}
			\scalebox{0.9}{$
				\tb{ \bm \theta} = (\underbrace{0.25,0.13,0.14,0.17}_{\text{imitators}},\underbrace{\textcolor{red}{0.11},\textcolor{magenta}{0.05},\textcolor{purple}{0.03}}_{\text{nonconformists}},\underbrace{\textcolor{blue}{0.12}}_{\substack{\text{conformists}}})^\top,
				$}
        \end{equation*}
        and utility functions
		\begin{small}
			\begin{alignat*}{2}
				\textcolor{red}{u^\1_1}(x) &= \tb{-1.3x +1}, \quad
				&& \textcolor{red}{u^\2_1}(x) = \tb{2.5x -1.3}, \\
				\textcolor{magenta}{u^\1_2}(x) &= \tb{-0.7x + 0.4},  \quad
				&& \textcolor{magenta}{u^\2_2}(x) = \tb{1.5x -0.3}, \\
				\textcolor{purple}{u^\1_3}(x) &= \tb{-0.8x+ 1}, \quad
				&&  \textcolor{purple}{u^\2_3}(x)=\tb{2x + 0.2}, \\
				\textcolor{blue}{u^\1_4}(x) &= \tb{0.7}, \quad
				&&  \textcolor{blue}{u^\2_4}(x)=\tb{-0.556x + 0.856}.
			\end{alignat*}
		\end{small}
		The corresponding thresholds are $\textcolor{red}{\tau_1} = \tb{0.605}$,  $\textcolor{purple}{\tau_2}=  \tb{0.318}$, $\textcolor{magenta}{\tau_3} = \tb{0.286}$, and $\textcolor{blue}{\tau'_1} = \tb{0.281}$.
		\tb{Starting from the origin (initially, no one is enrolled), a population of 100 individuals may exhibit persistent fluctuations in the enrollment proportion $x(k)$ (\Cref{fig_finitepop}, Panel a).
If fluctuations in the enrollment proportion persist as $\N$ grows, health officials should expect uncertainty in demand in large population sizes. 
As the population size increases, however, the amplitude of fluctuations in the enrollment proportion decreases (\Cref{fig_finitepop}--Panels b and c). 
So large-amplitude oscillations in real enrollment data would likely reflect exogenous drivers (e.g., incidence shocks, policy changes, or supply disruptions) rather than endogenous social learning alone.}
	\end{example}
	\begin{figure*}
		\centering
		\includegraphics[width=1\linewidth, height=0.5\linewidth,trim=2.0cm 3.3cm 7cm 0.85cm, clip]{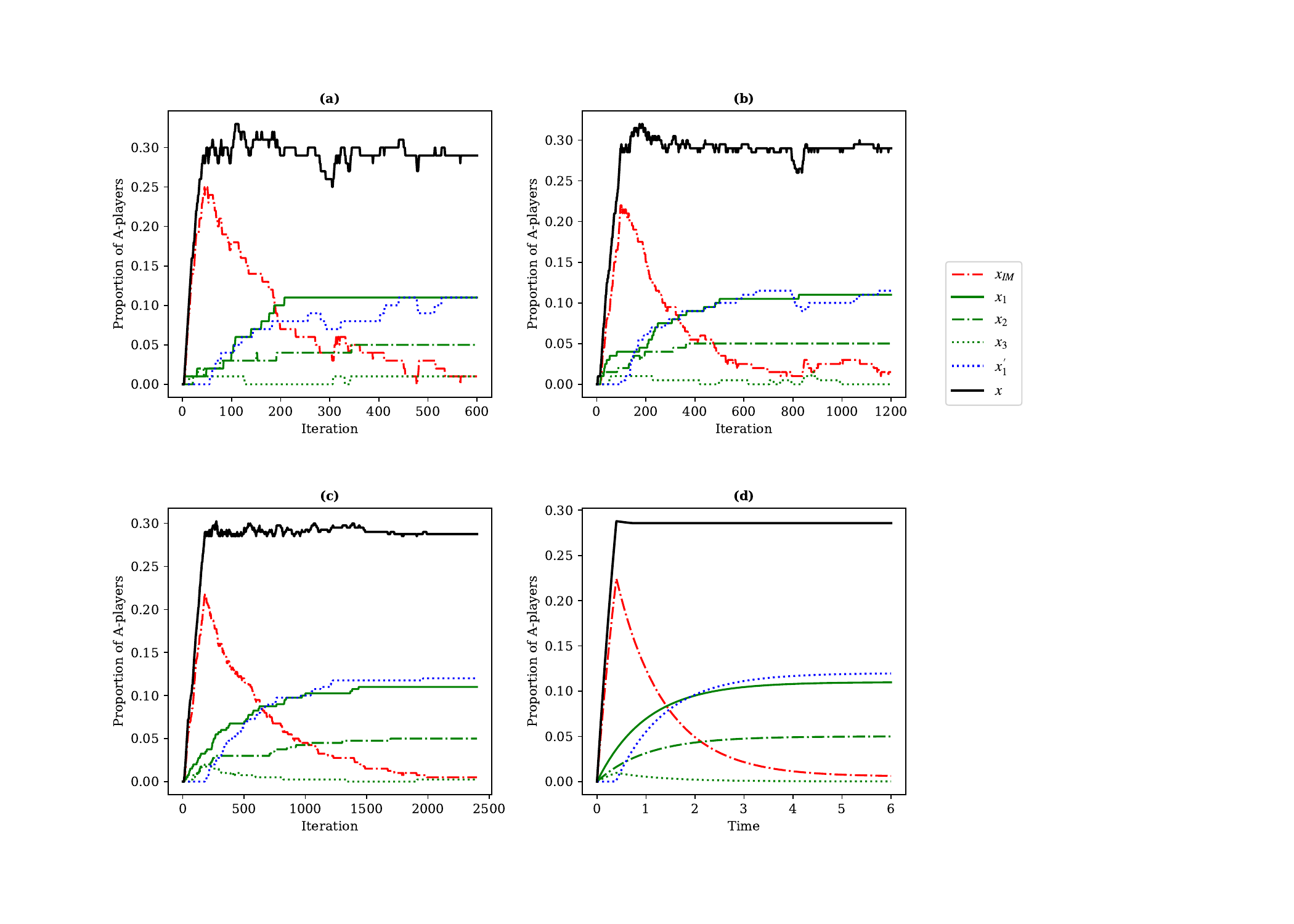}
		\caption{\textbf{Evolution of discrete mixed population dynamics for three population sizes in Example \ref{example} and that of the continuous-time.}
			Each panel shows the evolution of the population proportions over time, where $x_{IM}$ is the proportion of imitators playing strategy $\1$; $x_i$ for $i = 1,2,3$ (resp. $x'_1$) denotes the proportion of $\1$-playing nonconformists  (resp. conformists) of type $i$, and $x$ is the proportion of $\1$-players. 
			Panel (a) corresponds to  $\N = 100$  units over $600$ time steps; panels (b) and (c) are associated with $\N$ equal to $200$ and $400$, respectively. 
			As the $\N$ increases, the amplitude of fluctuations in $x$ (black curve) decreases.
			Panel (d) depicts the evolution of the corresponding continuous-time population dynamics.}
		\label{fig_finitepop}
	\end{figure*}
	
	In \Cref{example}, the amplitude of the fluctuations reduces with population size.
	Whether this observation generalizes to other mixed populations remains to be investigated.
	Stochastic approximation theory can help reveal the discrete dynamics' asymptotic behavior as the population size approaches infinity. 
	%PR: 
	We start with some background and definitions.
	
	\section{Background}
	A \emph{differential inclusion} on  $\mathbb{R}^n$ is defined by
	\begin{equation} \label{eq_differentialInclusion}
		\dot{\x} \in \bm{\mathcal{V}}(\x),
	\end{equation}
	where $\bm{\mathcal{V}}:\mathbb{R}^n \rightrightarrows \mathbb{R}^n$ is a set-valued map. A point $\x^* \in \mathbb{R}^n$ is called an \emph{equilibrium} if $\mb{0} \in \bm{\mathcal{V}}(\x^*)$.
	A \emph{selection} from $\bm{\mathcal{V}}(\x)$ is a singleton map $\bm{\nu}(\x) \in \mathbb{R}^n$ which satisfies $\bm{\nu}(\x) \in \bm{\mathcal{V}}(\x)$ for all $\x$ in $\mathbb{R}^n$ \cite{smirnov2002introduction}.
	
	A set $\bm{\mathcal{M}}$ is called \emph{attractive} from a set $\bm{\mathcal{U}}$ under differential inclusion $\bm{\mathcal{V}}$ if for each neighborhood  $\bm{\mathcal{M}}^\epsilon$ of set $\bm{\mathcal{M}}$ each solution trajectory starting from a point in $\bm{\mathcal{U}}$   enters and remains $\bm{\mathcal{M}}^\epsilon$ after some finite time \cite{MAYHEW20111045}.
	The union of all such sets $\bm{\mathcal{U}}$ forms the \emph{basin of attraction} of $\bm{\mathcal{M}}$ \cite{MAYHEW20111045}.
	
	The following definitions are adopted from \cite{sandholm2010population}.
	The \emph{tangent cone} to a closed convex set $\bm{\mathcal{X}}\subseteq \mathbb{R}^n$ at a point $\x \in \mathbb{R}^n$, denoted by $\mathrm{TC}(\x)$, is given by
	$$\mathrm{TC}(\x) = \text{cl}\left(\{\bm{z} \in \mathbb{R}^n : \bm{z} = \alpha(\y - \x), \alpha \geq 0, \y \in \bm{\mathcal{X}}\}\right),$$
	where $\text{cl}$ denote the closure of a set.
	The \emph{projection} of $\mathbb{R}^n$ onto $\bm{\mathcal{X}}$ is defined by $\Pi_{\bm{\mathcal{X}}}:\mathbb{R}^n\to\bm{\mathcal{X}}$, where
	$ 
	\Pi_{\bm{\mathcal{X}}}(\x) = \arg\min_{\y \in \bm{\mathcal{X}}}\|\y - \x\|.
	$
	Suppose the set-valued map $\bm{\mathcal{V}}:\bm{\mathcal{X}} \rightrightarrows \mathbb{R}^n$ satisfies: (i) $\bm{\mathcal{V}}(\x) \subseteq \mathrm{TC}(\x)$, (ii) $\bm{\mathcal{V}}(\x)$ is nonempty, convex, and bounded, and (iii) the graph of $\bm{\mathcal{V}}$, defined as $\{(\x,\y) \mid \y \in \bm{\mathcal{V}}(\x)\}$, is closed. 
	Extend $\bm{\mathcal{V}}$ to $\mathbb{R}^n$ by setting $\bm{\mathcal{V}}(\x) := \bm{\mathcal{V}}(\Pi_{\bm{\mathcal{X}}}(\x))$ for $\x \in \mathbb{R}^n \setminus \bm{\mathcal{X}}$. 
	Then the resulting differential inclusion \eqref{eq_differentialInclusion} for this extended $\bm{\mathcal{V}}$ is \emph{good upper semicontinuous}, and the set ${\bm{\mathcal{X}}}$ is forward invariant \cite[Theorem 6.A.2]{sandholm2010population}. 
	We use this projection-based extension throughout this paper.
	The following definitions are adopted from \cite{roth2013stochastic}. 
	We denote the set of solutions to the differential inclusion 
	$\dot{\x} \in \bm{\mathcal{V}}(\x)$ starting at $\x_0$ by $\bm{\mathcal{T}}_{\x_0}.$
	The set $\bm{\mathcal{T}}_{\x_0},$ is a subset of the space of continuous map from $\mathbb{R}_{\geq 0} \to \bm{\mathcal{X}}.$
	The set of all solutions to the differential inclusion 
	$\dot{\x} \in \bm{\mathcal{V}}(\x)$
	is denoted by $\bm{\mathcal{T}}_{\Phi}$. 
	The set-valued dynamical system induced by this differential inclusion is defined as $\bm{\Phi}:\mathbb{R}_+ \times \bm{\mathcal{X}} \rightrightarrows \mathbb{R}^n$, where $\bm{\Phi}_t(\x_0) = \{\x(t) \in \bm{\mathcal{T}}_{\x_0}\}$. The \emph{limit set} of a point $\x_0$, denoted by $\bm{\mathcal{L}}(\x_0)$, is
	$
	\bigcup_{\y \in \bm{\mathcal{T}}_{\x_0}} \bigcap_{t \geq 0}\text{cl}(\y[t, \infty]).
	$
	The set of \emph{recurrent points} is
	$
	\{\x_0 \mid \x_0 \in \bm{\mathcal{L}}(\x_0)\},
	$
	and the closure of the set of recurrent points is called the \emph{Birkhoff center} of $\bm{\Phi}$.
	For brevity, we drop ``the dynamical systems induced by.''
	
	\begin{definition}[\cite{roth2013stochastic}] \label{defGSAP}
		For a sequence of positive values $\epsilon$ approaching zero, let $\mathbi{U}^{\epsilon} = \langle \mathbi{U}_k^{\epsilon}\rangle_{k=0}^{\infty}$ and $\langle \bm{\mathcal{V}}^{\epsilon}\rangle$ be a sequence of random variables and a family of set-valued maps, respectively. Consider a good upper semicontinuous differential inclusion $\dot{\x} \in \bm{\mathcal{V}}(\x)$ over a compact, convex state space $\bm{\mathcal{X}}$. The sequence $\mb{X}^{\epsilon}= \langle \mb{X}_k^{\epsilon}\rangle_{k=0}^{\infty}$ is a \emph{generalized stochastic approximation process (GSAP)} for the differential inclusion if:
		\begin{enumerate}
			\item $\mb{X}_k^{\epsilon}\in \bm{\mathcal{X}}$ for all $k \geq 0$,
			\item $\mb{X}_{k+1}^{\epsilon}-\mb{X}_k^{\epsilon}-\epsilon\mathbi{U}_{k+1}^{\epsilon}\in \epsilon\bm{\mathcal{V}}^{\epsilon}(\mb{X}_k^{\epsilon})$,
			\item $\forall \delta >0 \exists \epsilon_0>0$ 
			$\forall \epsilon \leq \epsilon_0 \forall \x \in \bm{\mathcal{X}}$\\
            \vspace{-10pt}
            $$
				\bm{\mathcal{V}}^{\epsilon}(\x) \subset \{\bm{z} \in \mathbb{R}^n \mid \exists \y:\!\vert\x-\y\vert< \delta, \! 
				\inf_{\bm{v}\in \bm{\mathcal{V}}(\y)} \vert \bm{z} - \bm{v}\vert< \delta \},
            $$
            \vspace{-5pt}
			\item For all $T>0$ and all $\alpha>0$, the following holds uniformly in $\x \in \bm{\mathcal{X}}$: 
			$$
			\lim_{\epsilon \rightarrow 0} \mathbb{P} 
			\left[ \max_{k \leq \frac{T}{\epsilon}} \left \vert 
			\textstyle\sum_{i = 1}^{k} \epsilon \mathbi{U}^{\epsilon}_{i} \right \vert > \alpha \mid \mb{X}^{\epsilon}_0 = \x \right] = 0.
			$$
		\end{enumerate}
	\end{definition}
	\tr{
	For a discrete-time sequence
	$\langle  \mb{X}^{\frac{1}{\N}}_k\rangle_{k = 0}^{\infty}$
	defined over $\bm{\mathcal{X}}$, we define an associated affine
	\emph{interpolated process} $\bar{\bm{X}}^\frac{1}{\N}$ which is defined over continuous time $t$
	as follows
	$$
	\bar{\bm{X}}^\frac{1}{\N}(t) =  \mb{X}^{\frac{1}{\N}}_{l(t)} + \big(\N t - l(t) \big) \big(\mb{X}^{\frac{1}{\N}}_{l(t)+1} - \mb{X}^{\frac{1}{\N}}_{l(t)}\big),
	$$
	where $l(t) = \lfloor \N t \rfloor$.
	The next theorem relates the finite-horizon evolution of the GSAP to the corresponding differential inclusion as the population size approaches infinity.
	\begin{theorem} \cite{roth2013stochastic} \label{thm:shortTermSandholm}
		Let  $\mb{X}^{\epsilon}$  be a family of GSAPs for a differential inclusion 
		$\dot{\x} \in \bm{\mathcal{V}}(\x)$.
		Then for any $T>0$ and any $\alpha >0$, we have
		$$
		\lim_{\frac{1}{\N} \to 0} \mathbb{P}\! \left[ \!\inf_{\x \in \bm{\mathcal{T}}_{\Phi}} \sup_{0 \leq s \leq T} \left \vert \bar{\bm{X}}^{\frac{1}{\N}}(s) - \x(s) \right \vert  \geq \alpha \Bigm|\ \!\bar{\bm{X}}^{\frac{1}{\N}}(0) = \x_0 \right]\! = \!0
		$$
		uniformly in $\x_0 \in {\bm{\mathcal{X}}}$.
	\end{theorem}
	}
	As for asymptotic behavior, in \cite{roth2013stochastic}, it has been proved that
	if for each $\epsilon >0$ 
	the GSAP	$\langle \mathbf{X}^{\epsilon}_{k} \rangle_{k=0}^\infty$ is a Markov chain,	the support of  limit points of stationary measures in the topology of weak convergence, as $\epsilon$ approaches zero, will be contained in the Birkhoff center of the differential inclusion $\dot{\x} \in \bm{\V}(\x).$
	The following theorem is a result of \cite[Theorem 3.5 and Corollary 3.9]{roth2013stochastic}.
		\begin{theorem}[\cite{aghaeeyan2023discrete}] \label{thm:implicationOFSandholm}
		For the differential inclusion $\dot{\x} \in \bm{\mathcal{V}}(\x)$, let $\mb{X}^{\epsilon}$ be GSAP Markov chains with an invariant measure $\mu^{\epsilon}$ for a vanishing sequence $\epsilon$. 
		Then, for any open neighborhood $\bm{\mathcal{O}}$ of the Birkhoff center of the  differential inclusion, we have
		$
		\lim_{\epsilon\to 0}\mu^{\epsilon}(\bm{\mathcal{O}})=1.
		$
	\end{theorem}
	\section{The link to the continuous-time dynamics}
	In view of \Cref{thm:implicationOFSandholm}, if 
	\emph{(i)} the discrete mixed population dynamics form a Markov chain, 
	and \emph{(ii)} the family of the Markov chains indexed by the population size is a GSAP for a good upper-semicontinuous differential inclusion, then the asymptotic behavior of the discrete mixed population dynamics with a population size approaching infinity can be revealed by that of the differential inclusion.
	The mean dynamic is indeed a selection from this differential inclusion.
	
	To show that the discrete mixed population dynamics form a Markov chain, first, we write the mixed population dynamics described in \Cref{sec_problemFormulation}
	in a compact way.
	We define the function ${s}^*(p,\x): [2\pp]\times (\bm{\X}_{ss} \cap\frac{1}{\mathsf{N}}\mathbb{Z}^{{2\pp}}) \rightarrow \{1,2\}$ returning ``preferred strategy'' of a type-$p$ agent at the population state $\x$ which equals $1$ (resp. $2$) if $\1$ (resp. $\2$) is the preferred strategy.
	By the ``preferred strategy,'' we mean that dictated by the agent's update rule and will be adopted if the agent is active.
	Hence,
	\begin{alignat}{2} \label{eq_s_star}
		&{s}^*(p,\x) 
		= \\
		&\begin{cases}
			{1,}& \text{ if } (x \leq {\tau_{p-\pp}}  \text{ and } \pp < p \leq \pp + \p) \\
			&\text{ or }
			(x \geq \tau_{2\pp +1 -p}'  \text{ and } p > \pp+\p)    \\
			& \text{ or } \big(\displaystyle \max_{i\in[\pp]} u^\1_i(x)\mt{1}(x_i + x_{\pp + i})  \\
			& \qquad \quad \geq \displaystyle\max_{i\in[\pp]} u^\2_i(x)\mt{1}(\theta_i + \theta_{\pp+i} - x_i  -x_{\pp+i})  \\
			&\qquad \text{ and } p \leq \pp\big), \\
			{2,} & \text{otherwise,}
			\nonumber
		\end{cases}
	\end{alignat}
	Utility $u^\1_i(x)$ (resp. $u^\2_i(x)$), for $i \in [\pp]$,  is linked to both the imitative non-imitative type $i$, and is taken into account in determining the highest earners if some agents in either of these two types are playing strategy $\1$ (resp. $\2$), or, equivalently,
	either $x_i$ or $x_{\pp+i}$ is not zero (resp. is not equal to $\theta_i$ or $\theta_{\pp+i}$).
	
	\begin{proposition} \label{prop_discreteDynamics} 
		The discrete mixed population dynamics are equivalent to the dynamics implied by the following discrete time stochastic equation for $k\in\mathbb{Z}_{\geq0}$:
		\begin{align}\label{eq:discretePopulationDynamics}
			\x(k+1)  &= \x(k) + \frac{1}{\mathsf{N}} \big({S}_k-s^*(P_k,\x(k))\big)\mb{e}_{{P}_k},
		\end{align}
		where ${P}_k$ is a random variable with the support $[2\pp]$ and distribution $\mathbb{P}[{P}_k  = p] = \theta_p$ for $p \in [2\pp],$ random variable 
		${S}_k$  has the support $\{1,2\}$ and distribution  $\mathbb{P} [{S}_k  = 1 \vert P_k = p] = x_p/\theta_p, \mathbb{P} [{S}_k  = 2 \vert P_k = p] = 1-x_p/\theta_p$, and $\mb{e}_{{P}_k}$ is the $P_k$-th unit vector in $\mathbb{R}^{2\pp}.$
	\end{proposition}
	
	In \Cref{prop_discreteDynamics}, the random variable ${P}_k$ represents the type of the active agent at time index $k$, while the random variable ${S}_k$ takes the value $1$ (resp. $2$) if the active agent's strategy at index $k$ is $\1$ (resp. $\2$).
	So ${S}_k-s^*(P_k,\x(k))$ indicates the difference between the current and preferred strategy of the active agent and $\mb{e}_{{P}_k}$ positions the difference at the row associated with the type of the active agent in the population state $\x$.
	
	Given \eqref{eq:discretePopulationDynamics}, the next state of the population dynamics is determined by the current population state and the type and strategy of the currently active agent.
	Consequently, the sequence $\langle\x(k)\rangle_{k=0}^\infty$ is a realization of the Markov chain defined below, which appears to be the same as \cite[Definition 5]{aghaeeyan2023discrete} but with a different function $s^*$.
	\begin{definition} \label{def_Markov_chain}
		The \textbf{\emph{mixed population dynamics Markov chain}} \tb{for a population of size $\N$} is defined as the Markov chain $\langle \mathbf{X}^{\frac{1}{\mathsf{N}}}_k\rangle_{k=0}^{\infty}$ with 
		the state space $\bm{\mathcal{X}}_{ss} \cap \frac{1}{\mathsf{N}}\mathbb{Z}^{{2\pp}}$, initial state $\mathbf{X}^{\frac{1}{\N}}_0 = \x(0)$, and transition probabilities
		\begin{align}
			& \text{Pr}_{\x,\y} = \label{eq:markov}\\ 
			&\begin{cases} 
				({\theta}_p - x_p)(2-{s}^*(p,\x)), &\hspace{-39pt}\text{if }\exists p(\y = \x+ \frac{1}{\mathsf{N}}\mb{e}_p), \nonumber \\
				x_p (s^*(p,\x)-1), &\hspace{-39pt}\text{if } \exists p(\y =  \x -\frac{1}{\mathsf{N}}\mb{e}_p),
				\\
				1- \Big(\displaystyle\sum_{p=1}^{2\pp} x_p (s^*(p,\x)-1)&\hspace{-39pt}\text{if } \y=  \x,\\
				\hspace{15pt}  
				\quad \qquad  + ({\theta}_p - x_p)(2-s^*(p,\x))
				\Big),
				& \\
				0, &\hspace{-39pt}\text{otherwise},
			\end{cases}   
		\end{align}
	where $\mb{e}_p$ is the $p^\text{th}$ unit vector in  $\mathbb{R}^{2\pp}$. 
    	\end{definition}
	It can be easily shown that the sequence $\langle \x(k)\rangle_{k=0}^{\infty}$, where $\x(k)$ is governed by \eqref{eq:discretePopulationDynamics}, is a realization of $\langle\mathbf{X}^{\frac{1}{\mathsf{N}}}_k\rangle_{k=0}^{\infty}$.
	See \cite{aghaeeyan2023discrete} for a proof of a similar result.
	Note that \Cref{prop_discreteDynamics} and \Cref{def_Markov_chain} take the same form as that in \cite{aghaeeyan2023discrete}, but the difference is that the underlying dynamics and, hence, the random variables and the function $s^*$ are different.

	% We are now ready to define the differential inclusion for which the Markov chain defined above, indexed by the population size, is a GSAP (\Cref{lemGSAPS}).
	The mean dynamics \tb{give} the expected increment in the population proportion of $\1$-players in each type per time step, in the large-population limit, and can be written as the difference between the expected inflow of $\1$-players and the expected outflow of $\1$-players.
	The mean dynamic induced by the set-up introduced in previous sections is a selection from the following differential inclusion.
	\begin{definition}  \label{def_semicontinuousDynamics}
		The \textbf{\emph{continuous-time mixed population dynamics}} are defined by 
		\begin{subequations}
			\begin{align}
				\dot{\x}\in \bm{\mathcal{V}}(\x),
			\end{align}\label{eq:type_mixed}
		\end{subequations}
		where $\bm{\mathcal{V}} :\bm{\mathcal{X}}_{ss} \rightrightarrows \mathbb{R}^{2\pp}$ and for $0 < p \leq \pp $
		%\begin{subequations} 
		\begin{alignat*}{2}
			{\V}_p(\x) \!=\!
			\begin{cases} 
				\{{\theta}_p- x_p\}, &\text{if }  u^\1(x) >   u^\2(x) \\
				&\quad \text{and } u^\1(\x) >   u^\2(\x),  \\ 
				\{- x_p\},&\text{if }  u^\1(x) <   u^\2(x) 
				\\&\quad \text{and } u^\1(\x) <   u^\2(\x),  \\  
				[-x_p, {\theta}_p-x_p], & \text{otherwise},
			\end{cases}    \tag{\ref{eq:type_mixed}b}
		\end{alignat*}  
		%\label{eq:type_mixed_setValued_function}
		%\end{subequations}
		for $\pp < p \leq \pp + \p $,
		\begin{align*}
			{\V}_p(\x) \!= \!\begin{cases} 
				\{{\theta}_p - x_p\}, & \text{if } x < \tau_{p-\pp},   \\
				[-x_p, {\theta}_p-x_p], & \text{if }  x = \tau_{p-\pp}, \\
				\{- x_p\}, & \text{if } x > \tau_{p-\pp},   
			\end{cases}\tag{\ref{eq:type_mixed}c}
		\end{align*} 
		and for  $p > \pp+\p$ 
		\begin{align*}
			\V_p(\x)\! = \!\begin{cases}
				\! \{ - x_p \}, &\hspace{-5pt} \text{if } x < \tau'_{2\pp+1-p},  \\
				\!    [-x_p, {\theta}_{p}-x_p],  &\hspace{-5pt}\text{if } x = \tau'_{2\pp+1-p},\\
				\!   \{ {\theta}_{p} - x_p\},  &\hspace{-5pt}\text{if } x > \tau'_{2\pp+1-p}, 
			\end{cases}     \tag{\ref{eq:type_mixed}d}
		\end{align*}
		where  
		\begin{align*}
			x=\sum_{p=1}^{2\pp}x_p \tag{\ref{eq:type_mixed}e}
		\end{align*}
		is the proportion of $\1$-players in the population.
	\end{definition}
	
	Except for the first $\pp$ elements of $\bm{\mathcal{V}}$, which correspond to imitative types, 
	the remaining elements 
	are the same as those of the differential inclusion obtained for a population of purely best-responders \cite{aghaeeyan2023discrete}.
	The set-valued map $\V_p(\x)$ for $p \leq \pp$  is a singleton as long as 
	the maximum $\1$-utility  and active $\1$-utility are either both greater than or both less than maximum  $\2$-utility and active  $\2$-utility, respectively. 
	If, at some $\x$, the maximum $\1$-utility exceeds the maximum $\2$-utility while the active $\1$-utility is less than the active $\2$-utility (or vice versa), or if any of the $\1$-utilities equals its $\2$ counterpart, then $\V_p(\x)$ is non-singleton.

	We analyze the behavior of the discrete population dynamics \eqref{eq:discretePopulationDynamics} for increasing population size while the population structure, i.e., the vector $\bm{\theta}$ remains unchanged.
	Hence, the sequence $\langle N\rangle_{N = \N}^{\infty}$ along which the population size approaches infinity must satisfy $N  \bm{\theta} \in \mathbb{Z}^{2\pp}$.   
	The set of all such feasible population sizes is represented by $\mathcal{N}$.
	By  $\langle \frac{1}{\N}\rangle_{\N \in \mathcal{N}}$,  we mean a decreasing sequence $\frac{1}{\N}, \N \in \mathcal{N}$,  converging to zero.
	
	It can be shown that the differential inclusion \eqref{eq:type_mixed} satisfies the conditions of good upper semi-continuity.
	We then have the following lemma. 
	\begin{lemma}    \label{lemGSAPS}
		For vanishing sequence $\langle \frac{1}{\N}\rangle_{\N \in \mathcal{N}}$, 
		the collection of $\mb{X}^{\frac{1}{\mathsf{N}}} = \langle  \mb{X}^{\frac{1}{\mathsf{N}}}_k\rangle_{k=0}^{\infty}$  is a GSAP for \eqref{eq:type_mixed}.
	\end{lemma}
	
	\tr{ Given \Cref{lemGSAPS} and \Cref{thm:shortTermSandholm},
	the interpolated process $\bar{\bm{X}}^\frac{1}{\N}$ of the mixed population dynamics Markov chain, over any finite time interval, closely tracks some solution of the continuous-time population dynamics, with high probability, when population size approaches infinity.
	\begin{proposition} \label{prop_transient_dynamics}
		For any $T>0$ and for any $\alpha >0$ we have
		$$
		\lim_{\frac{1}{\N} \to 0}\! \mathbb{P} \!\left[ \!\inf_{\x(t) \in \bm{\mathcal{T}}_{\Phi}} \sup_{0 \leq s \leq T} \vert \bar{\bm{X}}^{\frac{1}{\N}}(s) - \x(s) \vert \geq \alpha\! \mid \bar{\bm{X}}^{\frac{1}{\N}}(0) \!=\! \x_0 \right]\! = \!0
		$$
		uniformly in $\x_0 \in {\bm{\mathcal{X}}_{ss}}$, where
		$\bar{\bm{X}}^\frac{1}{\N}$ is the interpolated process  of the mixed population dynamics Markov chain \eqref{eq:markov} and
		$\bm{\mathcal{T}}_{\Phi}$ is the set of all solutions of the continuous-time mixed population dynamics \eqref{eq:type_mixed}.
	\end{proposition}
	}
	What about the asymptotic behavior of the discrete mixed population dynamics?
	Given \Cref{lemGSAPS} and \Cref{thm:implicationOFSandholm}, the asymptotic behavior of the discrete mixed population dynamics, for a population size approaching infinity, can also be revealed by that of the differential inclusion defined in \eqref{eq:type_mixed}.
	\begin{remark} \label{remark_other_tie_breakers}
		The set-up introduced in \eqref{eq:imitation}, \eqref{eq:scor}, and \eqref{eq:santi} is based on the specific tie-breaking rule; agents choose strategy $\1$ whenever there is no difference between the strategies.
		However, the paper's main results remain valid when other tie-breaking rules are considered.
		More specifically, a different tie-breaker changes only $s^{*}(\cdot, \cdot)$ when a tie happens.
		This, in turn, changes the selection from the differential inclusion at the tie for which Condition 3 in \Cref{defGSAP} is satisfied.
		Thus, the Markov chain with transition probability \eqref{eq:markov} remains a GSAP for the differential inclusion \eqref{eq:type_mixed}.
	\end{remark}
	\section{Main results}
	%We first introduce additional definitions and assumptions necessary to present the main results.
	
	\tb{The following definition characterizes classes of population states; later, we show that the dynamics' equilibrium points are contained in these classes.}
	
	\begin{definition} \label{def_population-equilibrium-point}
		The population state $\x \in \bm {\mathcal{X}}_{ss}$ is 
		\begin{enumerate}
			\item  \emph{ clean-cut} if it is characterized by 
			\begin{align*}
				(\zeta_1, \ldots, \zeta_{\pp},
				\rho_1, \ldots, \rho_{i}, 0, \ldots, 0,\rho_j', \ldots,\rho_1'),
			\end{align*}
			where
			$\max \{\tau_{i+1}, \tau'_{j} \}< {\eta}'_{j}+ {\eta}_{i} + \zeta_0 < \min \{\tau_{i}, \tau'_{j+1} \}$ and $u^\1(\eta_i + \eta_j'+ \zeta_0) > u^\2(\eta_i + \eta_j'+ \zeta_0)$;
			\item   \emph{imitation-free clean-cut} if it is characterized by 
			$$(\underbrace{0, \ldots, 0}_{\pp}, \rho_1, \ldots, \rho_{i}, 0, \ldots, 0,\rho_j',  \ldots, \rho_1'),$$ where
			$\max \{\tau_{i+1}, \tau'_{j} \}< {\eta}'_{j}+ {\eta}_{i}  < \min \{\tau_{i}, \tau'_{j+1} \}$ and $u^\1(\eta_i + \eta_j') < u^\2(\eta_i + \eta_j')$;
			\item   \emph{mixed clean-cut} if it belongs to the following set 
			\begin{align}
				\bm{\mathcal{C}}^{c}_{i,j} = 
				\big\{\x &\in \bm{\mathcal{X}}_{ss}\mid  \sum_{p=1}^{\pp}x_p  = c - (\eta_{i} + \eta'_{j}), \nonumber\\
				& x_{ p} = \theta_{p}, \text{ if }\pp + 1 \leq\! p\! \leq \!\pp + i \nonumber\\
				& \qquad \qquad \text{ or } 2\pp -j + 1 \leq p \!\leq\! 2\pp,\nonumber\\
				&  x_{ p} = 0,  \text{ if } \pp + i \! < \!p \!<\! 2\pp -j + 1  \big\}, \label{eq_mixed_clean}
			\end{align}
			for some  
			$c \in (\max \{\tau_{i+1}, \tau'_{j} \}, \min \{ \tau_i, \tau'_{j+1}\})$, and $ i \in [\p], j \in [\p'], c \in \mathcal{C}$;
			\item  \emph{non-mixed nonconformist-driven} if it is characterized by
			\begin{align*}
				\big(r_1, \ldots,& r_{\pp}, 
				\rho_1, \ldots, \rho_{i-1}, \tau_i - (\sum_{p=1}^{\pp} r_i  + \eta_{i-1} + \eta_j'), 0,\\
				&\ldots,0, \rho_j',  \ldots, \rho_1'\big),
			\end{align*}
			where $ \tau_i \notin \mathcal{C}$, $j =\max \{k| \tau'_k < \tau_{i} \} $,
			\begin{equation*}
				r_p =
				\begin{cases}
					0, & \text{if } u^\1(\tau_i) < u^\2(\tau_i),\\
					\theta_p, & \text{if } u^\1(\tau_i) > u^\2(\tau_i),
				\end{cases}
			\end{equation*}
			for $p \in [\pp]$,  and
			\begin{equation*}
				\tau_i \in 
				\begin{cases}
					(\eta_{i-1} + \eta_j', \eta_{i} + \eta_j'), \\ \quad \qquad \text{if } u^\1(\tau_i) < u^\2(\tau_i),\\
					(\eta_{i-1} + \eta_j' + \zeta_0, \eta_{i} + \eta_j' + \zeta_0), \\
					\quad \qquad \text{if }
					u^\1(\tau_i) > u^\2(\tau_i);
				\end{cases}
			\end{equation*} 
			\item  \emph{mixed nonconformist-driven} if it belongs to the following set
			\begin{align}
				\bm{\mathcal{R}}^n_{i,j} = \big\{\x &\in\bm{\mathcal{X}}_{ss} \mid \sum_{p=1}^{\pp}x_p + x_{\pp + i} = \tau_i - (\eta_{i-1} + \eta'_{j}), \nonumber \\
				& x_{ p} = \theta_{ p}, \text{ if }\pp + 1 \leq p \leq \pp + i-1 \nonumber\\
				& \qquad \qquad \text{ or } 2\pp -j + 1 \leq p \leq 2\pp,\nonumber\\
				&  x_{ p} = 0, \! \text{ if } \pp + i \! < \!p \!< \!2\pp -j + 1 
				\big\},\label{eq_mixed_anti}
			\end{align}
			for some   $ i \in [\p]$,  $j =\max \{k| \tau'_k < \tau_{i} \} $,
			$\tau_i \in \mathcal{C} \cap
			(\eta_{i-1} + \eta_j', \eta_{i} + \eta_j' + \zeta_0);$ 
			\item  \emph{non-mixed conformist-driven} if it is characterized by
			\begin{align*}
				\big(r_1, \ldots, & r_{\pp}, 
				\rho_1, \ldots, \rho_{i}, 0, \ldots,0,\\
				&\tau_j' - (\sum_{p=1}^{\pp} r_i  + \eta_{i} + \eta_{j-1}'),
				\rho_{j-1}',\ldots,  \rho_1'\big),
			\end{align*}
			where $\tau_j' \notin \mathcal{C}$,
			$i= \max \{k \in [\p]| \tau'_j < \tau_{k} \} $,
			\begin{equation*}
				r_p  =
				\begin{cases}
					0, & \text{if } u^\1(\tau_j') < u^\2(\tau_j'),\\
					\theta_p, & \text{if } u^\1(\tau_j') > u^\2(\tau_j'),
					%\\
					% [0,\zeta_0], & \text{otherwise, }
				\end{cases}
			\end{equation*}
			for $p \in [\pp]$,  and
			\begin{equation*}
				\tau'_j \in 
				\begin{cases}
					(\eta_{i} + \eta_{j-1}', \eta_{i} + \eta_j'), \\
					\quad \qquad \text{if } u^\1(\tau_j') < u^\2(\tau_j'),\\
					(\eta_{i} + \eta_{j-1}'+\zeta_0, \eta_{i} + \eta_j' + \zeta_0),
					\\
					\quad \qquad \text{if }
					u^\1(\tau_j') > u^\2(\tau_j'),
				\end{cases}
			\end{equation*}
			\item \emph{mixed conformist-driven} if it belongs to the following set 
			\begin{align}
				\bm{\mathcal{R}}^c_{i,j}\! =\!
				\big\{\!\x& \in\bm{\mathcal{X}}_{ss} \!\mid \sum_{p=1}^{\pp}x_p + x_{2\pp -j+1}\! =\! \tau'_j - (\eta_{i} + \eta'_{j-1}), \nonumber\\
				& x_{ p} = \theta_{ p}, \text{ if }\pp + 1 \leq p \leq \pp + i \nonumber\\
				& \qquad \qquad \text{ or } 2\pp -j + 2 \leq p \leq 2\pp,\nonumber\\
				&  x_{ p} = 0,  \text{ if } \pp + i \! < p <\! 2\pp -j + 1
				\big\},\label{eq_mixed_cor}
			\end{align}
			for some $ j \in [\p']$, $i= \max \{k \in [\p]| \tau'_j < \tau_{k} \} $,
			$\tau'_j \in \mathcal{C} \cap 
			(\eta_{i} + \eta_{j-1}', \eta_{i} + \eta_j' + \zeta_0)$.
		\end{enumerate}
		\tb{In the population states defined above,  $i=0$ (resp. $j = 0$) implies that no nonconformists (resp. conformists) are playing $\1$.}
	\end{definition}
	\begin{remark} \label{remark_continuum}
		Three types of population states—mixed nonconformist-driven, mixed conformist-driven, and mixed clean-cut—form a continuum.
	\end{remark}
	
	The following lemma characterizes the equilibria of the continuous-time population dynamics \eqref{eq:type_mixed}.
	
	\begin{lemma} \label{lem_equilibriumPointofContinuous}
		Let $\bm{\mathcal{Q}}$ be the set of all equilibrium points of the continuous-time population dynamics.
		Under Assumptions \ref{ass:ass1}-\ref{ass:ass3}, 
		%  each equilibrium point  of the
		% continuous-time population dynamics \eqref{eq:type_mixed}
		% corresponds to one of the five population states defined in \Cref{def_population-equilibrium-point}.
		\tb{every member of the set}  $\bm{\mathcal{Q}}$ \tb{is one of the } population states defined in \Cref{def_population-equilibrium-point}.
	\end{lemma}
	\tb{The proof of \Cref{lem_equilibriumPointofContinuous} is provided in the appendix and based on the definition of an equilibrium point of differential inclusions, and that at an equilibrium point, zero should belong to each of the  $2\pp$ elements of the set-valued map described in \eqref{eq:type_mixed}.}
	% Denote the set of all equilibrium points of the continuous-time population dynamics by $\bm{\mathcal{Q}}.$
	We now present the main result of this paper.
	\begin{theorem} \label{thm_birkhoff_center_invariant_population}
		For a population of $\N$ agents playing two-strategy games governed by \eqref{eq:discretePopulationDynamics}, let $\mu^{\frac{1}{\N}}$ be an invariant probability measure for the corresponding mixed population dynamics Markov chain.
		Under Assumptions \ref{ass:ass1}-\ref{ass:ass3},
		for any vanishing sequence   $\langle \frac{1}{\N} \rangle_{\N \in \mathcal{N}}$  we have $\lim_{ \frac{1}{\N} \to 0} {\mu}^{\frac{1}{\N}}(\bm{\mathcal{O}})=1,$
		where $\bm{\mathcal{O}}$ is any
		open set containing the set $\bm{\mathcal{Q}}$.
		% characterized in \Cref{lem_equilibriumPointofContinuous}.
	\end{theorem}
	
	In words, consider a population of size $\N$ consisting of imitators, conformists, and nonconformists evolving according to discrete population dynamics \eqref{eq:discretePopulationDynamics}.
	As the population size grows to infinity and in the long term, with probability one, the population state visits the points close to the equilibria.
	Accordingly, the amplitude of the perpetual fluctuations in the population proportion of $\1$-players approaches zero in probability. 
	This results in the following corollary.
	%, whose proof is similar to \cite[Corollary 4]{aghaeeyan2023discrete}.
	\begin{corr} \label{corollary}
		Under the conditions of \Cref{thm_birkhoff_center_invariant_population},  the amplitude of the fluctuations in the population proportions of $\1$-player converges to zero with probability one.
	\end{corr}
	The population dynamics may still experience perpetual fluctuations even for large population sizes.
	However, most of the time, the population dynamics remain close to equilibria.
	\begin{remark}
		\tb{
			This paper extends our previous work \cite{Hien2} from finite populations—where agents played against the entire population, including themselves—to a large population size ($\N \to \infty$). 
			We keep the same structure to enable direct comparison. 
			In large populations, whether an agent includes themselves in the population state is often irrelevant, and decisions depend primarily on action prevalence (e.g., 20\% play $\1$ and 80\% play $\2$). 
			In our earlier conformist/nonconformist model, excluding oneself led to the same qualitative result: fluctuations vanish almost surely \cite{aghaeeyan2023discrete}.
		}
	\end{remark}
	\section{Supporting Results}
	According to \Cref{thm:implicationOFSandholm}, as the population size grows to infinity, the support of the invariant probability measures for the population dynamics Markov chain is contained within any open set that includes the Birkhoff center of the differential inclusion \eqref{eq:type_mixed}. 
	On the other hand, in \Cref{thm_birkhoff_center_invariant_population}, it is claimed that this support is contained within any open set that includes the equilibrium points of \eqref{eq:type_mixed}, hinting that the Birkhoff center and the set of equilibrium points are equal.
	In this section, we show that this conjecture is true.
	
	In \cite{aghaeeyan2023discrete}, the analysis of the population dynamics was facilitated by that of one-dimensional dynamics, which captured the evolution of the population proportion of $\1$-players.
	Motivated by this, we calculate the derivative of the population proportion of $\1$-players (\ref{eq:type_mixed}e), which is given in  \eqref{eq:populationProportion}.
	Unlike the case in \cite{aghaeeyan2023discrete}, the differential inclusion \eqref{eq:populationProportion} is not necessarily one-dimensional.
	The dynamics, however, will be one-dimensional over any finite time horizon, as given in \eqref{eq:abstract_het}, if the initial condition $\x(0)$ lies in the interior of the state space, namely, $\x(0) \notin \partial \bm{\mathcal{X}}_{ss}$.
	\tb{In this case, the maximum $\1$-utility and the maximum active $\1$-utility are equal. 
		Likewise, the maximum $\2$-utility and the maximum active $\2$-utility are equal. 
		Thus, in \eqref{eq:populationProportion}, terms with state $\x$ will be inactive.}
	
	What happens when the initial condition is on the boundary of the state space?
	In this case, the asymptotic behavior of the population proportion of $\1$-players is captured by the one-dimensional dynamics in \eqref{eq:abstract_het}, provided both start from the same initial condition (\Cref{lem_convergence_of_population_proportion} and \Cref{basin-of-attraction-of-abstract-dynamics}).
	\tb{We thus define the \emph{abstract}  state and the corresponding dynamics as follows. }
	\begin{definition} \label{prop2}
		The \emph{abstract dynamics} associated with the continuous-time population dynamics are defined by
		\eqref{eq:abstract_het},
		\begin{figure*}
			\begin{gather} 
				\dot{\hat{x}} \in {\mathcal{X}}(\hat{x}), \label{eq:abstract_het}\\
				\scalebox{0.87}{$
					\mathcal{X}(\hat{x}) = \begin{cases}
						[\eta_{i-1} + \eta'_{j} + \zeta_0-\hat{x}, \eta_{i} + \eta'_{j} + \zeta_0-\hat{x}],  &\text{if }  \exists i (\hat{x} = \tau_{i}) \text{ and }
						u^\1 (\hat{x}) > u^\2(\hat{x}),
						\\
						[\eta_{i} + \eta'_{j-1} + \zeta_0-\hat{x}, \eta_{i} + \eta'_{j} + \zeta_0-\hat{x}],  &\text{if }  \exists j (\hat{x} = \tau_{j}') \text{ and }
						u^\1 (\hat{x}) > u^\2(\hat{x}),\\
						[\eta_{i-1} + \eta'_{j}-\hat{x}, \eta_{i} + \eta'_{j}-\hat{x}],  &\text{if }  \exists i (\hat{x} = \tau_{i}) \text{ and }
						u^\1 (\hat{x}) < u^\2(\hat{x}),
						\\
						[\eta_{i} + \eta'_{j-1}-\hat{x}, \eta_{i} + \eta'_{j}-\hat{x}],  &\text{if }  \exists j (\hat{x} = \tau_{j}') \text{ and }
						u^\1 (\hat{x}) < u^\2(\hat{x}),\\
						[\eta_{i-1} + \eta'_{j}-\hat{x}, \eta_{i} + \eta'_{j}+ \zeta_0-\hat{x}],  &\text{if }  \exists i (\hat{x} = \tau_{i}) \text{ and }
						u^\1 (\hat{x}) = u^\2(\hat{x}),
						\\
						[\eta_{i} + \eta'_{j-1}-\hat{x}, \eta_{i} + \eta'_{j}+\zeta_0-\hat{x}],  &\text{if }  \exists j (\hat{x} = \tau_{j}') \text{ and }
						u^\1 (\hat{x}) = u^\2(\hat{x}),\\
						[\eta_{i} + \eta'_{j} -\hat{x}, \eta_{i} + \eta'_{j} + \zeta_0 -\hat{x}],      &\text{if }  \exists i,j (\hat{x} \in   (\max \{\tau'_{j}, \tau_{i+1}\}, \{\tau'_{j+1}, \tau_{i} \})) \text{ and }
						u^\1 (\hat{x}) = u^\2(\hat{x}),
						\\   
						\{\eta_{i} + \eta'_{j}+\zeta_0 -\hat{x}\}, &\text{if }  \exists i,j (\hat{x} \in   (\max \{\tau'_{j}, \tau_{i+1}\}, \{\tau'_{j+1}, \tau_{i} \})) \text{ and }
						u^\1 (\hat{x}) > u^\2(\hat{x}),
						\\
						\{\eta_{i} + \eta'_{j} -\hat{x}\},    &\text{if }  \exists i,j (\hat{x} \in   (\max \{\tau'_{j}, \tau_{i+1}\}, \{\tau'_{j+1}, \tau_{i} \})) \text{ and }
						u^\1 (\hat{x}) < u^\2(\hat{x}),
					\end{cases}
					$ } \nonumber
			\end{gather}
		\end{figure*}
		where
		at $\hat{x} = \tau_i$, $j$ is defined by $\max \{k| \tau'_k < \tau_{i} \}$, and
		at  $\hat{x} = \tau_j$,  $i$ is defined by $\max \{k| \tau'_j < \tau_{k} \}$.
		The scalar state $\hat{x}$ is \tb{called} the \emph{abstract state}.
		%\end{equation*}
	\end{definition}
	Following the terminology in \Cref{def_population-equilibrium-point},
	we call an abstract state $\hat{x}\! \in\! [0,1]$ a nonconformist-driven (resp. conformist-driven) state if $\hat{x}$ equals $\tau_p$ for some $p \in [\p]$ (resp. $\tau'_p$ for some $p \in [\p']$).
	The abstract state $\hat{x} \in [0,1]$ is clean-cut (resp. imitation-free clean-cut)  if there exists some $i \in \{0\} \cup [\p], j \in \{0\} \cup [\p']$ such that
	$\hat{x} =\eta_i + \eta_j' + \zeta_0$ (resp. $\hat{x} =\eta_i + \eta_j'$).
	Finally, the abstract state $\hat{x} \in [0,1]$ is a mixed clean-cut state if there exists some $i \in \{0\} \cup [\p], j \in \{0\} \cup [\p']$, and $ r \in (0, \zeta_0)$ such that
	$\hat{x} =\eta_i + \eta_j' + r$.
	
	The following lemma characterizes the equilibrium points of the abstract dynamics.
	
	\begin{lemma} \label{lem_abstractEquilibrium}
		Under Assumptions \ref{ass:ass1}-\ref{ass:ass3}, 
		each equilibrium point of the abstract dynamics \eqref{eq:abstract_het} falls exclusively into one of the following categories:
		\begin{enumerate}
			\item nonconformist-driven equilibrium point, $\tau_i$, $i \in  [\p]$, if and only if one of the following conditions holds
			\begin{itemize}
				\item  $u^\1(\tau_i) < u^\2(\tau_i)$ and  $\tau_i \in (\eta_{i-1} + \eta_j', \eta_{i} + \eta_j'),$ 
				\item  $\tau_i \in \mathcal{C}$ and  $\tau_i \in (\eta_{i-1} + \eta_j', \eta_{i} + \eta_j' + \zeta_0),$ 
				\item $u^\1(\tau_i) > u^\2(\tau_i)$ and  $\tau_i \in (\eta_{i-1} + \eta_j' + \zeta_0, \eta_{i} + \eta_j' + \zeta_0);$ 
			\end{itemize}
			\item conformist-driven equilibrium point, $\tau_j'$, $j \in [\p']$, if and only if one of the following conditions holds
			\begin{itemize}
				\item  $u^\1(\tau_j') < u^\2(\tau_j')$ and  $\tau_j' \in (\eta_{i} + \eta_{j-1}', \eta_{i} + \eta_j'),$ 
				\item  $\tau_j' \in \mathcal{C}$ and  $\tau_j' \in (\eta_{i} + \eta_{j-1}', \eta_{i} + \eta_j' + \zeta_0),$
				\item $u^\1(\tau_j') > u^\2(\tau_j')$ and  $\tau_j' \in (\eta_{i} + \eta_{j-1}' + \zeta_0, \eta_{i} + \eta_j' + \zeta_0);$  
			\end{itemize}
			\item imitation-free clean-cut equilibrium point $\eta_i + \eta_j'$, $i \in \{0\} \cup [\p]$, $j \in \{0\} \cup [\p']$, if and only if  
			$\max \{\tau_{i+1}, \tau'_{j} \}< {\eta}'_{j}+ {\eta}_{i} < \min \{\tau_{i}, \tau'_{j+1} \}$ and $u^\1(\eta_i + \eta_j') < u^\2(\eta_i + \eta_j');$
			\item  clean-cut equilibrium point $\eta_i + \eta_j'+ \zeta_0$, $i \in \{0\} \cup [\p]$, $j \in  \{0\} \cup [\p']$, if and only if  
			$\max \{\tau_{i+1}, \tau'_{j} \}< {\eta}'_{j}+ {\eta}_{i} + \zeta_0 < \min \{\tau_{i}, \tau'_{j+1} \}$ and $u^\1(\eta_i + \eta_j'+\zeta_0) > u^\2(\eta_i + \eta_j'+\zeta_0);$
			\item mixed clean-cut equilibrium point $\eta_i + \eta_j'+ r$, $i \in \{0\} \cup [\p]$, $j \in \{0\} \cup [\p']$,  $r\in (0,\zeta_0),$ if and only if  
			$\max \{\tau_{i+1}, \tau'_{j} \}< {\eta}'_{j}+ {\eta}_{i} + r < \min \{\tau_{i}, \tau'_{j+1} \}$ and $(\eta_i + \eta_j'+ r) \in \mathcal{C}.$
		\end{enumerate}
	\end{lemma}
	\tb{The mutual exclusiveness of the five categories of the abstract states in \Cref{lem_abstractEquilibrium} is immediate from \Cref{ass:ass1}, and the fact that the thresholds associated with types are distinct.}
	The following lemma characterizes the stability properties of the abstract equilibrium points.
	
	% Checked Pouria: Write this in terms of the payoff values. Currently, it's almost the same as restating the definition of stability
	\begin{lemma} \label{lem_instability_of_mixed_clean_cut}
		Consider the abstract dynamics \eqref{eq:abstract_het}.
		Under Assumptions \ref{ass:ass1}-\ref{ass:ass3}, 
		the set of unstable equilibrium points of the abstract dynamics consists of conformist-driven equilibrium points and those mixed clean-cut equilibria $q$ satisfying
		$b_{i(q)} < d_{i(q)}$, 
		%  $u^\1_{i(q)}(0) <  u^\2_{j(q)}(0)$,
		where
		$i(q)$ (resp. $j(q)$) is the type with the greatest $\1$-utility (resp. $\2$-utility) at $x = q$.
		The remaining equilibrium points are asymptotically stable.
	\end{lemma}
	
	In \Cref{lem_instability_of_mixed_clean_cut}, for a mixed clean-cut equilibrium point, the inequality $b_{i(q)} < d_{i(q)}$ implies that the maximum $\2$-utility exceeds the maximum $\1$-utility in the left neighborhood of $q$, rendering the equilibrium point unstable.
	%$u^\1_{i(q)}(0) <  u^\2_{j(q)}(0)$
	
	It is easy to show that the abstract dynamics admit at least one asymptotically stable equilibrium point. 
	Sort the asymptotically stable equilibrium points in ascending order 
	$
	q^*_1 < q^*_2 < \ldots < q^*_{\Q -1} < q^*_{\Q},
	$
	where $\Q$ is the number of stable equilibrium points.
	It can be shown that between any two consecutive stable equilibrium points, $q^*_{k}, q^*_{k+1}$, for $k \in [\Q-1]$, there is one unstable equilibrium point, denoted by $q^*_{k,k+1}.$
	See \cite[Lemma A1]{aghaeeyan2023discrete} for proof of a similar result.
	The following proposition reveals the
	basin of attractions of the abstract equilibria.
	\begin{proposition} \label{basin-of-attraction-of-abstract-dynamics}
		Consider the abstract dynamics \eqref{eq:abstract_het}.
		Under Assumptions \ref{ass:ass1}-\ref{ass:ass3},
		the basin of attraction of the equilibrium point $q^*_k$, $k \in [\Q],$ of the abstract dynamics equals
		\begin{equation}
			\mathcal{A}(q^*_k) = 
			\begin{cases}
				[0, q^*_{1,2}), & \text{if } k=1, \\
				(q^*_{k-1,k},q^*_{k,k+1}), & \text{if } 2 \leq k \leq \Q-1, \\
				(q^*_{\Q-1,\Q},1], & \text{if } k=\Q. 
			\end{cases}
		\end{equation}    
	\end{proposition}
	So far, we have analyzed the global convergence of the abstract dynamics, showing that each solution approaches one of the equilibrium points characterized in \Cref{lem_abstractEquilibrium}.
	We now turn to the global convergence of the population dynamics.
	The following lemma reveals the relation between the equilibria of the abstract dynamics and those of \eqref{eq:type_mixed}.
\begin{lemma} \label{lemma_correspondece_abstract_population_eq}
	Under Assumptions \ref{ass:ass1}-\ref{ass:ass3}, the following hold:
	\begin{enumerate}
		\item For each abstract equilibrium point $q$, there exists at least one equilibrium point $\bm q \in \bm{\mathcal{Q}}$ of the dynamics \eqref{eq:type_mixed} such that $\sum_{p=1}^{2\pp} q_p = q$, where $q_p$ denotes the $p^\text{th}$ element of $\bm q$.
		\item Conversely, each equilibrium point $\bm q \in \bm{\mathcal{Q}}$ of the dynamics \eqref{eq:type_mixed} corresponds to exactly one abstract equilibrium point $q$ satisfying $\sum_{p=1}^{2\pp} q_p = q$.
	\end{enumerate}
\end{lemma}
% 
%     The proof is similar to that of \cite[Lemma 2-Part 2]{aghaeeyan2023discrete} and is omitted.
% \end{proofref}

% By \Cref{lem_equilibriumPointofContinuous}, each equilibrium point of the continuous-time population dynamics corresponds to exactly one of the population states defined in \Cref{def_population-equilibrium-point} which, in turn,
%  corresponds to an abstract state defined in \Cref{def_abstract-equilibrium-point}.
% Given this and noting that nonconformist-driven and conformist-driven states for $r\in (0, \rho_i)$ form a continuum of states,
% we conclude that the continuous-time population dynamics admit at least as many equilibrium points as the abstract dynamics.

The following result characterizes the global convergence behavior of the continuous-time population dynamics. 

\begin{theorem} \label{lem:attraction-population-single-equilibrium}
For the continuous-time population dynamics \eqref{eq:type_mixed},
under Assumptions \ref{ass:ass1}-\ref{ass:ass3}, the following hold:
\begin{enumerate}
	\item Any clean-cut, imitation-free clean-cut, or non-mixed nonconformist-driven 
	equilibrium point $\bm q\in \bm{\mathcal{Q}}$, if any,  is attractive
	with the basin of attraction
	$\{\x \in \bm{\mathcal{X}}_{ss} \mid \bm 1^\top \x \in  \mathcal{A}(\bm 1^\top \bm q) \}.$
	\item The continuum of mixed nonconformist-driven 
	equilibrium points $\bm{\mathcal{R}}^n_{i,j}$, $ i \in [\p]$, $ j = \max \{k \leq \p'\mid \tau'_k < \tau_i \}$,  if any, is attractive
	with the basin of attraction
	$ \{\x \in \bm{\mathcal{X}}_{ss} \mid \bm 1^\top \x \in  \mathcal{A}(\bm 1^\top \bm q) \text{ for some }\bm q\in \bm{\mathcal{R}}^n_{i,j} \}.$
	\item The continuum of mixed clean-cut 
	equilibrium points $ \bm{\mathcal{C}}^{[c]}_{i,j}$, $ i \in [\p], j \in [\p']$, if any,  is attractive if and only if its corresponding abstract equilibrium point $\bm 1^\top \bm q$, for $\bm q \in \bm{\mathcal{C}}^{[c]}_{i,j}$, is asymptotically stable.
	In this case, the set
	$\{\x \in \bm{\mathcal{X}}_{ss} \mid \bm 1^\top \x \in  \mathcal{A}(\bm 1^\top \bm q) \text{ for some } \bm q\in \bm{\mathcal{C}}^{[c]}_{i,j} \}$
	is the basin of attraction of the continuum.
	\item The limit set of every point in the set $\{\x\in \bm{\X}_{ss} \mid \bm 1^\top \x = {q}^*_{k,k+1} \}$ for $ k \in [\mathtt{Q}-1]$ is   
	$\{ \x \in \bm{\mathcal{Q}}  \mid \bm 1^\top \bm x \in \{q^*_{k}, q^*_{k,k+1}, q^*_{k+1}\} \}.$
\end{enumerate}
\end{theorem}
\begin{remark} \label{rem_structured_pop}
\tb{Population structure and its role in the asymptotic behavior of populations were also studied by scholars. 
In evolutionary games on a graph, the long-term behaviour might depend on the payoff matrix, the strategy updating rule, and the population structure \cite{lieberman2005evolutionary}.
Despite this, the evolution of the expected proportion of each strategy in a large population of imitators, arranged on a regular graph, can be approximated by a replicator differential equation (RFE), analogous to the equation describing evolution in a large well-mixed population \cite{ohtsuki2006replicator}.
When game dynamics are coupled with environmental feedback, the asymptotic behavior inferred from the approximated RFE  can depend on the number of neighbors each player has \cite{stella2022lower}.
In this paper, we considered finite-size well-mixed  heterogeneous populations and investigated the fate of fluctuations for large enough population sizes.
%Our main contribution is to identify the deterministic dynamics underlying finite-population dynamics and to analyze their asymptotic behavior, thereby gaining insight into the asymptotic behavior of the corresponding discrete dynamics. 
Similar steps can be taken for a structured version of the studied populations.
This remains for future investigations.
}
\end{remark}
\begin{remark} \label{remak_contradictionOfAssumption2}

Without \Cref{ass:ass1}, \tb{the dynamics may admit new equilibria with structures the same as (imitation-free) clean-cut equilibria when the proportion of $\1$-players is in $\mathcal{C} \cup \mathcal{T}$.
Any new equilibrium at which} the population proportion of 
$\1$-players equals a conformist threshold is unstable.
If instead it matches an element of $\mathcal{C}$, its stability depends on the relative magnitudes of $u^\1_i(0)$ and $u^\2_j(0)$, where $i$ (resp. $j$) indexes the type with the greatest $\1$-utility (resp. $\2$-utility) at that point.
If \Cref{ass:ass2} does not hold, the stability of the abstract conformist and nonconformist equilibria depends on the highest-earning strategy near the abstract equilibrium point.
Accordingly, a conformist (resp. nonconformist) equilibrium point might be asymptotically stable or unstable.
In the absence of \Cref{ass:ass3}, the $\1$-utility and $\2$-utility lines of  different types may coincide.
If these utilities are the highest over some range of the population proportion of $\1$-players, imitators may fail to equilibrate.
More specifically, the evolution of $x_p$ for $p \in [\pp]$ follows $\dot{x}_p \in [-x_p, \theta_p - x_p]$, under which some trajectories exist that cause the proportion of $\1$-players among imitators and hence the population proportion of $\1$-players to fluctuate indefinitely.
Accordingly, the proportion of $\1$-players will fluctuate indefinitely.
\end{remark}
\begin{revisit}
% We re-ran the simulation with population size scaled from $68$ by successive doublings, up to $2176$, under $200$ different random activation sequences for each population size.
% The fluctuations in the population proportion of $\1$-players vanished for population sizes $136$ and greater, and the 
% population proportion of $\1$-players remained fixed at $x = \tau_1$.
% However, the proportions of $\1$-players within the nonconformists with threshold $\tau_1$  and the imitators varied across activation sequences, while their sum remained fixed at $\frac{16.5}{68}$.
The  associated continuous-time population dynamics with the discrete population introduced in \Cref{example} admit continua of nonconformist-driven equilibrium points characterized as 
    $$
	(x^*_1,x^*_2,x^*_3,x^*_4,0.11, 0.05,x^*_5,0.12),
	$$
where the sum of $x^*_p$, for $p\in[6]$, equals $\tau_3 - \eta'_1 \tb{- \eta_2} = 0.006$, and the population proportion of $\1$-players equals $\tau_3 \approx 0.286.$
Indeed, the population proportion of $\1$-players in finite discrete populations fluctuates around $\tau_3$. Similarly, the sum of the proportions of $\1$-playing imitators and nonconformistsf type 3 fluctuates around $0.006$ (\Cref{fig_finitepop}).

\end{revisit}
\section{Concluding Remarks}
We studied the asymptotic behavior of a well-mixed population of imitators, conformists, and nonconformists playing two-strategy games.
We derived and analyzed the mean dynamics of such a population in the form of a differential inclusion.
It was shown that the mean dynamics always converge to an equilibrium point or a continuum of equilibria.
Using the stochastic approximation theory, we showed that the amplitudes of the reported fluctuations in the proportion of $\1$-players in the population converge to zero with probability one (Figure \ref{fig_diag}).
The results of this paper, along with those of \cite{aghaeeyan2023discrete,aghaeeyanCDC}, suggest that in two-strategy, real-world decision-making problems where the proportions of adopters undergo large-amplitude fluctuations, individuals’ utility cannot be fully captured by affine functions of the proportion of adopters.

% \usetikzlibrary{positioning, shapes, arrows.meta, fit}
\begin{figure*}
\centering
\includegraphics[width=1\linewidth]{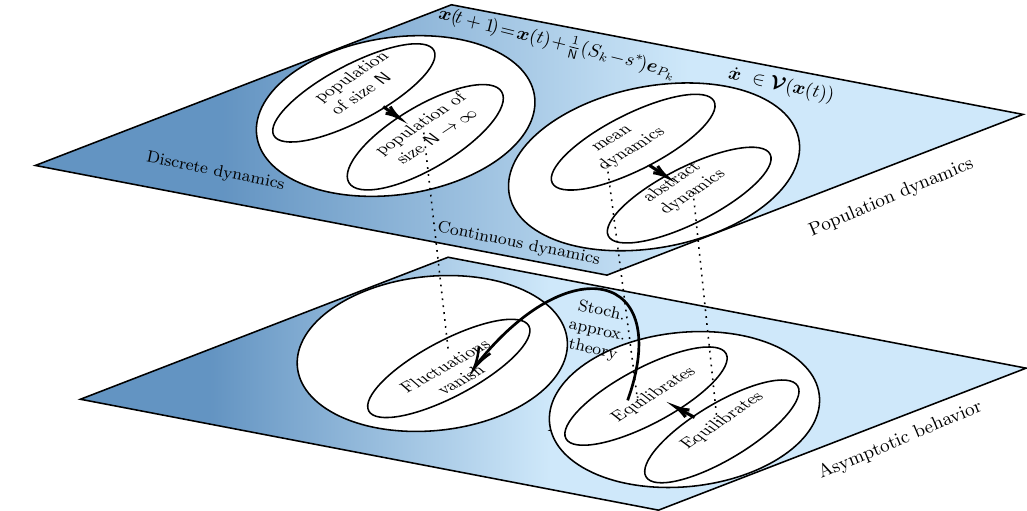}
\caption{\textbf{Stochastic approximation theory provides a link between the behavior of finite populations governed by the discrete-time dynamics and the corresponding continuous-time system.} 
	Analysis of the one-dimensional abstract dynamics facilitates that of the mean dynamics. 
	Convergence of the abstract dynamics implies convergence of the mean dynamics, which, in turn, implies that the discrete population dynamics converge in probability.}
\label{fig_diag}
\end{figure*}
% \section*{Declaration of Generative AI and AI-assisted technologies in the writing process}
% AA acknowledges partial assistance from ChatGPT for language polishing of the paper.
% \bibliographystyle{IEEEtran}
\bibliographystyle{plain}        % Include this if you use bibtex 
\bibliography{autosam}           % and a bib file to produce the 

@book{sandholm2010population,
  title={Population games and evolutionary dynamics},
  author={Sandholm, William H},
  year={2010},
  publisher={MIT press}
}

@ARTICLE{9234634,
  author={Como, Giacomo and Fagnani, Fabio and Zino, Lorenzo},
  journal={IEEE Transactions on Control of Network Systems}, 
  title={Imitation Dynamics in Population Games on Community Networks}, 
  year={2021},
  volume={8},
  number={1},
  pages={65-76},
  doi={10.1109/TCNS.2020.3032873}}

@article{roth2013stochastic,
  title={Stochastic approximations with constant step size and differential inclusions},
  author={Roth, Gr{\'e}gory and Sandholm, William H},
  journal={SIAM Journal on Control and Optimization},
  volume={51},
  number={1},
  pages={525--555},
  year={2013},
  publisher={SIAM}
}

@article{ZHU2023110707,
title = {Equilibrium analysis and incentive-based control of the anticoordinating networked game dynamics},
journal = {Automatica},
volume = {147},
pages = {110707},
year = {2023},
issn = {0005-1098},
author = {Yuying Zhu and Zhipeng Zhang and Chengyi Xia and Zengqiang Chen}
}

@article{granovetter1978threshold,
  title={Threshold models of collective behavior},
  author={Granovetter, Mark},
  journal={American journal of sociology},
  volume={83},
  number={6},
  pages={1420--1443},
  year={1978},
  publisher={University of Chicago Press}
}

@article{bauch2004vaccination,
  title={Vaccination and the theory of games},
  author={Bauch, Chris T and Earn, David JD},
  journal={Proceedings of the National Academy of Sciences},
  volume={101},
  number={36},
  pages={13391--13394},
  year={2004},
  publisher={National Acad Sciences}
}

@ARTICLE{comoImitation,
  author={Como, Giacomo and Fagnani, Fabio and Zino, Lorenzo},
  journal={IEEE Transactions on Control of Network Systems}, 
  title={Imitation Dynamics in Population Games on Community Networks}, 
  year={2021},
  volume={8},
  number={1},
  pages={65-76},
  doi={10.1109/TCNS.2020.3032873}}

@article{bestresponsePotential,
author = {Swenson, Brian and Murray, Ryan and Kar, Soummya},
title = {On Best-Response Dynamics in Potential Games},
journal = {SIAM Journal on Control and Optimization},
volume = {56},
number = {4},
pages = {2734-2767},
year = {2018}}

@article{arefizadeh2023robustness,
  title={Robustness of dynamics in games: A contraction mapping decomposition approach},
  author={Arefizadeh, Sina and Arefizadeh, Sadegh and Etesami, S Rasoul and Bolouki, Sadegh},
  journal={Automatica},
  volume={155},
  pages={111142},
  year={2023},
  publisher={Elsevier}
}

@ARTICLE{replicator,
  author={Madeo, Dario and Mocenni, Chiara},
  journal={IEEE Transactions on Automatic Control}, 
  title={Game Interactions and Dynamics on Networked Populations}, 
  year={2015},
  volume={60},
  number={7},
  pages={1801-1810},
  doi={10.1109/TAC.2014.2384755}}

@ARTICLE{replicator2,
  author={Mabrok, Mohamed A.},
  journal={IEEE Transactions on Automatic Control}, 
  title={Passivity Analysis of Replicator Dynamics and Its Variations}, 
  year={2021},
  volume={66},
  number={8},
  pages={3879-3884},
  doi={10.1109/TAC.2020.3027644}}

@book{smirnov2002introduction,
  title={Introduction to the Theory of Differential Inclusions},
  author={Smirnov, Georgi V},
  volume={41},
  year={2002},
  publisher={American Mathematical Soc.}
}

@INPROCEEDINGS{arditt2,
  author={Arditti, Laura and Como, Giacomo and Fagnani, Fabio and Vanelli, Martina},
  booktitle={2021 60th IEEE Conference on Decision and Control (CDC)}, 
  title={Equilibria and learning dynamics in mixed network coordination/anti-coordination games}, 
  year={2021},
  volume={},
  number={},
  pages={4982-4987},
  doi={10.1109/CDC45484.2021.9683414}}

@ARTICLE{coordinationandanticoordination,
  author={Como, Giacomo and Durand, Stéphane and Fagnani, Fabio},
  journal={IEEE Transactions on Automatic Control}, 
  title={Optimal Targeting in Super-Modular Games}, 
  year={2022},
  volume={67},
  number={12},
  pages={6366-6380},
  doi={10.1109/TAC.2021.3129733}}

@article{MAYHEW20111045,
title = {On the topological structure of attraction basins for differential inclusions},
journal = {Systems \& Control Letters},
volume = {60},
number = {12},
pages = {1045-1050},
year = {2011},
author = {Christopher G. Mayhew and Andrew R. Teel},
}

@article{aghaeeyan2023discrete,
author={Aghaeeyan, Azadeh and Ramazi, Pouria},
  journal={IEEE Transactions on Automatic Control}, 
  title={From Discrete to Continuous Binary Best-Response Dynamics: Discrete Fluctuations Almost Surely Vanish with Population Size}, 
  year={2025},
  volume={},
  number={},
  pages={1-16},
note={Early Access}
}

@article{fu2024evolutionary,
  title={Evolutionary matrix-game dynamics under imitation in heterogeneous populations},
  author={Fu, Yiheng and Ramazi, Pouria},
  journal={Automatica},
  volume={159},
  pages={111354},
  year={2024},
  publisher={Elsevier}
}

@INPROCEEDINGS{aghaeeyanCDC,
  author={Aghaeeyan, Azadeh and Ramazi, Pouria},
  booktitle={2024 IEEE 63rd Conference on Decision and Control (CDC)}, 
  title={From Discrete to Continuous Imitation Dynamics}, 
  year={2024},
  volume={},
  number={},
  pages={966-971},
  doi={10.1109/CDC56724.2024.10886827}}

@ARTICLE{Hien2,
  author={Le, Hien and Aghaeeyan, A and Ramazi, Pouria},
  journal={IEEE Transactions on Automatic Control}, 
  title={Heterogeneous Mixed Populations of Conformists, Nonconformists, and Imitators}, 
  year={2024},
  volume={69},
  number={5},
  pages={3373-3380}}

@article{zhang2018fashion,
  title={Fashion and homophily},
  author={Zhang, Boyu and Cao, Zhigang and Qin, Cheng-Zhong and Yang, Xiaoguang},
  journal={Operations Research},
  volume={66},
  number={6},
  pages={1486--1497},
  year={2018},
  publisher={INFORMS}
}

@article{kawagoe2023asymmetric,
  title={Asymmetric volunteer's dilemma game: Theory and experiment},
  author={Kawagoe, Toshiji and Takizawa, Hirokazu and Yamamori, Tetsuo},
  journal={Games and Economic Behavior},
  volume={142},
  pages={955--977},
  year={2023},
  publisher={Elsevier}
}

@article{chapman2012using,
  title={Using game theory to examine incentives in influenza vaccination behavior},
  author={Chapman, Gretchen B and Li, Meng and Vietri, Jeffrey and Ibuka, Yoko and Thomas, David and Yoon, Haewon and Galvani, Alison P},
  journal={Psychological science},
  volume={23},
  number={9},
  pages={1008--1015},
  year={2012},
  publisher={Sage Publications Sage CA: Los Angeles, CA}
}

@article{haslegrave2017majority,
  title={Majority dynamics with one nonconformist},
  author={Haslegrave, John and Cannings, Chris},
  journal={Discrete Applied Mathematics},
  volume={219},
  pages={32--39},
  year={2017},
  publisher={Elsevier}
}

@ARTICLE{7258336,
  author={Gharesifard, Bahman and Touri, Behrouz and Başar, Tamer and Shamma, Jeff},
  journal={IEEE Transactions on Automatic Control}, 
  title={On the Convergence of Piecewise Linear Strategic Interaction Dynamics on Networks}, 
  year={2016},
  volume={61},
  number={6},
  pages={1682-1687},
  doi={10.1109/TAC.2015.2477975}}

@ARTICLE{10453658,
  author={Arditti, Laura and Como, Giacomo and Fagnani, Fabio and Vanelli, Martina},
  journal={IEEE Transactions on Automatic Control}, 
  title={Robust Coordination of Linear Threshold Dynamics on Directed Weighted Networks}, 
  year={2024},
  volume={69},
  number={10},
  pages={6515-6529},
  doi={10.1109/TAC.2024.3371882}}

@article{kaniovski2000adaptive,
  title={Adaptive dynamics in games played by heterogeneous populations},
  author={Kaniovski, Yuri M and Kryazhimskii, Arkadii V and Young, H Peyton},
  journal={Games and Economic Behavior},
  volume={31},
  number={1},
  pages={50--96},
  year={2000},
  publisher={Elsevier}
}

@article{zino2025equilibrium,
  title={Equilibrium Selection in Replicator Equations Using Adaptive-Gain Control},
  author={Zino, Lorenzo and Ye, Mengbin and Calafiore, Giuseppe C and Rizzo, Alessandro},
  journal={IEEE Transactions on Automatic Control},
  year={2025},
  publisher={IEEE}
}

@article{barreiro2018constrained,
  title={Constrained evolutionary games by using a mixture of imitation dynamics},
  author={Barreiro-Gomez, Julian and Tembine, Hamidou},
  journal={Automatica},
  volume={97},
  pages={254--262},
  year={2018},
  publisher={Elsevier}
}

@ARTICLE{10172285,
  author={Chen, Ge and Yu, Yongyuan},
  journal={IEEE Transactions on Automatic Control}, 
  title={Convergence Analysis and Strategy Control of Evolutionary Games With Imitation Rule on Toroidal Grid}, 
  year={2023},
  volume={68},
  number={12},
  pages={8185-8192},
  doi={10.1109/TAC.2023.3291957}}

@ARTICLE{7458850,
  author={Riehl, James R. and Cao, Ming},
  journal={IEEE Transactions on Automatic Control}, 
  title={Towards Optimal Control of Evolutionary Games on Networks}, 
  year={2017},
  volume={62},
  number={1},
  pages={458-462},
  doi={10.1109/TAC.2016.2558290}}

@ARTICLE{7438812,
  author={Tan, Shaolin and Wang, Yaonan and Lü, Jinhu},
  journal={IEEE Transactions on Automatic Control}, 
  title={Analysis and Control of Networked Game Dynamics via A Microscopic Deterministic Approach}, 
  year={2016},
  volume={61},
  number={12},
  pages={4118-4124},
  doi={10.1109/TAC.2016.2545106}}

@ARTICLE{10416815,
  author={Liu, Aixin and Wang, Lin and Chen, Guanrong and Guan, Xinping},
  journal={IEEE Transactions on Cybernetics}, 
  title={Heterogeneously Networked Evolutionary Games With Intergroup Conflicts}, 
  year={2024},
  volume={54},
  number={10},
  pages={5684-5695},
  doi={10.1109/TCYB.2024.3352928}}

@article{lieberman2005evolutionary,
  title={Evolutionary dynamics on graphs},
  author={Lieberman, Erez and Hauert, Christoph and Nowak, Martin A},
  journal={Nature},
  volume={433},
  number={7023},
  pages={312--316},
  year={2005},
  publisher={Nature Publishing Group UK London}
}

@article{ohtsuki2006replicator,
  title={The replicator equation on graphs},
  author={Ohtsuki, Hisashi and Nowak, Martin A},
  journal={Journal of theoretical biology},
  volume={243},
  number={1},
  pages={86--97},
  year={2006},
  publisher={Elsevier}
}

@article{stella2022lower,
  title={Lower network degrees promote cooperation in the prisoner’s dilemma with environmental feedback},
  author={Stella, Leonardo and Baar, Wouter and Bauso, Dario},
  journal={IEEE Control Systems Letters},
  volume={6},
  pages={2725--2730},
  year={2022},
  publisher={IEEE}
}

@article{nowak1992evolutionary,
  title={Evolutionary games and spatial chaos},
  author={Nowak, Martin A and May, Robert M},
  journal={nature},
  volume={359},
  number={6398},
  pages={826--829},
  year={1992},
  publisher={Nature Publishing Group UK London}
}

@article{henrich2001evolution,
  title={The evolution of prestige: Freely conferred deference as a mechanism for enhancing the benefits of cultural transmission},
  author={Henrich, Joseph and Gil-White, Francisco J},
  journal={Evolution and human behavior},
  volume={22},
  number={3},
  pages={165--196},
  year={2001},
  publisher={Elsevier}
}

@article{grabisch2019model,
  title={A model of anonymous influence with anti-conformist agents},
  author={Grabisch, Michel and Poindron, Alexis and Rusinowska, Agnieszka},
  journal={Journal of Economic Dynamics and Control},
  volume={109},
  pages={103773},
  year={2019},
  publisher={Elsevier}
}

@article{galam2004contrarian,
  title={Contrarian deterministic effects on opinion dynamics:“the hung elections scenario”},
  author={Galam, Serge},
  journal={Physica A: Statistical Mechanics and its Applications},
  volume={333},
  pages={453--460},
  year={2004},
  publisher={Elsevier}
}

@article{BarrettEtAl2017,
  title   = {Pay-off-biased Social Learning Underlies the Diffusion of Novel Extractive Foraging Traditions in a Wild Primate},
  author  = {Barrett, Brendan J. and McElreath, Richard L. and Perry, Susan E.},
  journal = {Proceedings of the Royal Society B: Biological Sciences},
  year    = {2017},
  volume  = {284},
  number  = {1856},
  pages   = {20170358},
  doi     = {10.1098/rspb.2017.0358}
}

@article{BonoWhitenEtAl2018,
  title   = {Payoff- and Sex-Biased Social Learning Interact in a Wild Primate Population},
  author  = {Bono, Axelle E. J. and Whiten, Andrew and van Schaik, Carel P. and Kr{\"u}tzen, Michael and Graber, Stefanie M. and van de Waal, Erica},
  journal = {Current Biology},
  year    = {2018},
  volume  = {28},
  number  = {17},
  pages   = {2800--2805.e4},
  doi     = {10.1016/j.cub.2018.06.015}
}
% bibliography (preferred). The
% correct style is generated by
% Elsevier at the time of printing.

%\begin{thebibliography}{99}     % Otherwise use the  
% thebibliography environment.
% Insert the full references here.
% See a recent issue of Automatica 
% for the style.
%  \bibitem[Heritage, 1992]{Heritage:92}
%     (1992) {\it The American Heritage. 
%     Dictionary of the American Language.}
%     Houghton Mifflin Company.
%  \bibitem[Able, 1956]{Abl:56}
%     B.~C.~Able (1956). Nucleic acid content of macroscope. 
%     {\it Nature 2}, 7--9. 
%  \bibitem[Able {\em et al.}, 1954]{AbTaRu:54}   
%     B.~C. Able, R.~A. Tagg, and M.~Rush (1954).
%     Enzyme-catalyzed cellular transanimations.
%     In A.~F.~Round, editor, 
%     {\it Advances in Enzymology Vol. 2} (125--247). 
%     New York, Academic Press.
%  \bibitem[R.~Keohane, 1958]{Keo:58}
%     R.~Keohane (1958).
%     {\it Power and Interdependence: 
%     World Politics in Transition.}
%     Boston, Little, Brown \& Co.
%  \bibitem[Powers, 1985]{Pow:85}
%     T.~Powers (1985).
%     Is there a way out?
%     {\it Harpers, June 1985}, 35--47.

%\end{thebibliography}

\appendix
\renewcommand{\theequation}{A.\arabic{equation}}
\setcounter{equation}{0}
\section*{Appendix}
\begin{lemmaAppendix} \label{lem_appendix_2}
Consider the continuous-time population dynamics \eqref{eq:type_mixed}.
Under Assumptions \ref{ass:ass1}-\ref{ass:ass3},
the evolution of the population proportion of $\1$-players, $\sum_{p=1}^{2\pp}x_p(t)$, is governed by \eqref{eq:populationProportion}, where
$\x(t)$ is the continuous-time population state governed by \eqref{eq:type_mixed} and
for $x = \tau_i$, $j=\max \{k| \tau'_k < \tau_{i} \}$, and  for $x = \tau_j$, $i =\max \{k| \tau'_j < \tau_{k} \}$.
\end{lemmaAppendix}
\begin{proof}
The proof is based on the fact that the derivative of the population proportion of $\1$-players equals the sum of the derivative of the proportion of $\1$-players in each type, i.e., $\dot{\Tilde{x}} = \sum_{p=1}^{2\pp}\dot{x}_p(t)$.
The next steps are similar to those taken in the proof of Proposition 2 in \cite{aghaeeyan2023discrete} and are therefore omitted.
\end{proof}

\begin{lemmaAppendix} \label{lemmaAppendix}
Let $ \bm{\mathcal{X}}_1$ and $ \bm{\mathcal{X}}_2$ be defined as \eqref{eq:proof-equilibrium-sgn}.
\begin{figure*}
	\begin{subequations} 
		\label{eq:proof-equilibrium-sgn}
		\begin{align}
			\bm{\mathcal{X}}_1 &=   \{\x^* \in \bm{\mathcal{X}}_{ss} \mid  \bm1^\top \x^* \notin \mathcal{C} \text{ and } \nonumber\\
            & \qquad 
			\operatorname{sgn}(u^\1(\bm 1^\top \x^*) -u^\2(\bm 1^\top \x^*)) \neq \operatorname{sgn}(u^\1(\x^*) -u^\2(\x^*)) 
			\}, \label{eq:proof-equilibrium-sgn-X1}
			\\ \nonumber
			\bm{\mathcal{X}}_2 &=  \{\x^* \in \bm{\mathcal{X}}_{ss} \mid  
			\bm1^\top \x^* \in (\mathcal{C} -  \mathcal{T})
			\text{ and }
            \\ & \qquad 
			\operatorname{sgn}(u^\1(\bm 1^\top \x^*) -u^\2(\bm 1^\top \x^*)) \neq \operatorname{sgn}(u^\1(\x^*) -u^\2(\x^*))  \}. \label{eq:proof-equilibrium-sgn-X2}
		\end{align}
	\end{subequations}
\end{figure*}
Then     $ \bm{\mathcal{X}}_1 \cap \bm{\mathcal{Q}} = \varnothing$
and  $ \bm{\mathcal{X}}_2 \cap \bm{\mathcal{Q}} = \varnothing.$
\end{lemmaAppendix}
\begin{proof}
Assume on the contrary that there exists some $\x^* \in \bm{\mathcal{Q}} \cap \bm{\mathcal{X}}_1$.
$\bm 1^\top \x^* \notin \mathcal{C}$ implies
$u^\1(x^*) \neq u^\2(x^*)$, where $x^* = \bm 1^\top \x^*.$
We show a contradiction for the case $u^\1(x^*) >u^\2(x^*)$; the case $u^\1(x^*) < u^\2(x^*)$ can be handled similarly.
Without loss of generality, assume that $u^\1(x^*)$ is equal to $ u^\1_p(x^*)$ for some $p \in [\pp]$.
\tb{ We know $ u^\1(x^*) \geq  u^\1(\x^*)$ and  $u^\2(x^*) \geq  u^\2(\x^*)$
	Given $ u^\2(\x^*)  \geq u^\1(\x^*)$ and $ u^\1(x^*) > u^\2(x^*)$, we have $u^\1(\x^*) < u^\1(x^*) = u^\1_p(x^*)$. 
	This implies that $ x^*_p + x^*_{\pp + p}  = 0$, and consequently, $x^*_{\pp + p} = 0$.}
On the other hand,
the element  $x_{\pp + p}$ corresponds to a non-imitative type, \tb{and the two inequalities  
	$u^\2(x^*) \geq u^\2_p(x^*)$ and $u^\1(x^*) > u^\2(x^*)$ imply $u^\1_p(x^*)\! > \!u^\2_p(x^*)$, and, in turn, $\mathcal{V}_{\pp + p} \!= \!\{ \!\theta_p - x^*_p \!\}.$}
%     \begin{equation*}.
	%     \scalebox{0.9}{$
		%     \left.
		% \begin{aligned}
			% u^\1(x^*)& > u^\2(x^*)\\
			% u^\2(x^*) &\geq u^\2_p(x^*)
			% \end{aligned}
		% \right \rbrace
		%  \!\implies \!u^\1_p(x^*)\! > \!u^\2_p(x^*) \!\implies \!\mathcal{V}_{\pp + p} \!= \!\{ \!\theta_p - x^*_p \!\}.
		%  $}
	%     \end{equation*}  
%     \begin{equation}\scalebox{0.9}{$
		% \begin{aligned}
			%   \left.
			%   \begin{array}{r@{}l}
				%         u^\1(x^*) &= u^\1_p(x^*)\\
				%             u^\1(x^*)& > u^\2(x^*)\\
				%             u^\1(x^*) &\geq  u^\1(\x^*) \\
				%             u^\2(x^*) &\geq  u^\2(\x^*)\\
				%              u^\2(\x^*) & \geq u^\1(\x^*)
				%     \end{array} \right\rbrace
			%       \begin{array}{r@{}l}
				%           &\implies  u^\1(\x^*) < u^\1_p(x^*)\\
				%           & \quad \qquad \implies x^*_p + x^*_{\pp + p}  = 0  \\
				%           & \qquad \qquad \qquad \implies  x^*_{\pp + p}  = 0.
				%           \end{array}
			% \end{aligned}
		% $}
	% \end{equation}
% The case $u^\1(x^*) > u^\2(x^*)$ results in  $\max\{u^\1(x^*),  u^\2(x^*)\}$ equals $ u^\1_p(x^*)$ for some $p \in [\pp]$.
% $\x^* \in \bm{\mathcal{X}}_1$ implies
% $u^\1(\x^*) \leq u^\2(\x^*)$.
% With this and the fact that 
% $u^\2(\x^*) \leq u^\2(x^*)$,
% we have 
% $u^\1(\x^*) < u^\1(x^*)$ and, in turn, 
%  $u^\1(\x^*) < u^\1_p(x^*)$.
%  In view of \eqref{eq:UmaxAB-A}, this implies $\mt{1}(x_{p} + x_{\pp + p} ) = - \infty \implies x_{\pp + p} = 0$.
% indicates that 
%    the evolution of the element $x_{\pp + p}$ must be equal to $\rho_{\pp + p} - x_{\pp + p}$ and, in turn, $0 \notin \mathcal{V}_{\pp + p}(\x^*).$
Hence, $0 \notin \V_{\pp + p}(\x^*)$.
But the point $\x^*$ is an equilibrium.
This implies a contradiction.
Now, assume on the contrary that there exists some $\x^* \in \bm{\mathcal{Q}} \cap \bm{\mathcal{{X}}}_2$. 
The relation
$\bm 1^\top \x^* \in \mathcal{C}$ implies
$u^\1(x^*) = u^\2(x^*)$.
\tb{It holds that } $\exists p \in [\pp]\big( (u^\1(x^*) = u^\1_p(x^*))\big)$ and 
$exists q \in [\pp]\big( (u^\2(x^*) = u^\2_q(x^*))\big)$. 
\tb{Given} $x^* \notin \mathcal{T}$, \tb {$p$ and $q$ are not the same.}
% \begin{equation}     \scalebox{0.9}{$  \left.
		%     \begin{aligned}
			%         &\exists p \in [\pp]\big( (u^\1(x^*) = u^\1_p(x^*))\big)\\
			%           &\exists q \in [\pp]\big( (u^\2(x^*) = u^\2_q(x^*))\big)
			%     \end{aligned}
		%       \right \rbrace
		%        \xRightarrow{ x^* \notin \mathcal{T}} p \neq q$}
	% \end{equation}
On the other hand, \eqref{eq:proof-equilibrium-sgn-X2} implies 
that either
$u^\1(\x^*) > u^\2(\x^*)$ or  $u^\1(\x^*)< u^\2(\x^*)$.
We take the case where   $u^\1(\x^*) > u^\2(\x^*),$ and the other case can be handled similarly.
\tb{The inequalities $ u^\1(x^*) \geq  u^\1(\x^*) $, $ u^\2(x^*) \geq  u^\2(\x^*) $, $ u^\1(x^*) = u^\2(x^*)$ and  $u^\1(\x^*) > u^\2(\x^*)$ implies $ u^\2(\x^*) < u^\2(x^*)$, and, in turn, $ u^\2(\x^*) < u^\2_q(x^*)$.   }
% \[ \scalebox{0.9}{$\left.
	%  \begin{array}{r@{}l}
		%         u^\1(x^*) &= u^\1_p(x^*)\\
		%         u^\2(x^*) &= u^\2_q(x^*)\\
		%             u^\1(x^*)& = u^\2(x^*)\\
		%             u^\1(x^*) &\geq  u^\1(\x^*) \\
		%             u^\2(x^*) &\geq  u^\2(\x^*)\\
		%              u^\1(\x^*) &> u^\2(\x^*)
		%         \end{array} \right\rbrace
	%        % \begin{aligned}
		%           \implies  u^\2(\x^*) < u^\2_q(x^*), $}
	% \]
	So, we have
	$ \theta_q + \theta_{\pp + q} -x^*_{q} - x^*_{\pp+q} = 0 $ and
	$x^*_{\pp+q} = \theta_{\pp + q}$.
	Since $p \neq q$, it holds that $u^\1_q(x^*) < u^\2_q(x^*)$ and consequently $\V_{\pp+q}(\x^*) = \{-x_{q+\pp} \} \neq \{0\}$.
	But the point $\x^*$ is an equilibrium.
	This implies a contradiction.
\end{proof}
% For the continuous-time population dynamics \eqref{eq:type_mixed} and the corresponding abstract dynamics \eqref{eq:abstract_het},
% \begin{remark}
	%     When both dynamics \eqref{eq:abstract_het} and \eqref{eq:populationProportion} start from the same initial condition, the sets of trajectories of the two systems are the same as long as the relation \eqref{equal_sgn} holds.
	% Nevertheless, the evolution of the population proportion of 
	% $\1$-players follows a particular selection of the differential inclusion \eqref{eq:populationProportion}.
	% Accordingly, in the \Cref{lem_convergence_of_population_proportion}, we consider the evolution of the abstract state, which coincides with that of the population proportion of 
	% $\1$-players.
	% \end{remark}
% \begin{remark}
	%    Some parts of the proof \Cref{lem_convergence_of_population_proportion} rely on the fact that for any finite time, the state $\x$ does not belong to the boundary, i.e., $\x \in \bm{\mathcal{X}}_{ss}$.
	%    However, in real-world applications, no population size is infinite, and accordingly, when the population proportion of $\1$-players in some types becomes smaller than $1/\N$, the ``ottoeffect'' might happen.
	% \end{remark}
\begin{proofref}{\Cref{lemGSAPS}}
	The proof of \Cref{lemGSAPS} is inspired by the steps taken in \cite[Example 4.1]{roth2013stochastic} and \cite[Lemma 1]{aghaeeyan2023discrete}.
	Condition 1 is satisfied because the set $\bm{\mathcal{X}}_{ss}$ is convex and compact. 
	Now,
	consider the sequence $\langle \frac{1}{\N}\rangle_{\N\in\mathcal{N}}$.
	To prove that the collection $\langle \bm{X}^{\frac{1}{\N}}\rangle_{k=0}^{\infty}$  is a GSAP for the differential inclusion \eqref{eq:type_mixed}, we need to find the sequence
	$\langle \mb{U}^{\frac{1}{\N}}_k\rangle_{k}$ and
	$\langle\bm{\mathcal{V}}^{\frac{1}{\N}}\rangle$ satisfying the last three conditions listed in \Cref{defGSAP}.
	As for Condition 2, we consider the expected increment per time unit of  Markov chain $\bm{X}^{\frac{1}{\N}}$, denoted by $\bm{\nu}^{\frac{1}{\N}}$, as a candidate for $\bm{\mathcal{V}}^{\frac{1}{\N}}$.
	Each time unit is divided into $\N$ subunits.
	Hence, $\bm{\nu}^{\frac{1}{\N}}$ is equal to 
	$\bm{\nu}^{\frac{1}{\N}}(\bm{x}) = \smash{\N \mathbb{E}[\mathbf{X}^{\frac{1}{\N}}_{k+1} - \mathbf{X}^{\frac{1}{\N}}_{k} \mid \mathbf{X}^{\frac{1}{\N}}_k = \x]}$.
	The $p^{\text{th}}$ element of $\bm{\nu}^{\frac{1}{\N}}$, $p \in [2\pp]$, is the expected increment in type $p$ and equals the sum of the products of possible changes in type $p$ and their probabilities.
	Given \eqref{eq:markov}, two  changes might happen, namely, $\frac{1}{\N}$  increment or $\frac{1}{\N}$ reduction.
	Hence, ${\nu}^{\frac{1}{\N}}_p({\x})$ 
	is equal to
	$\theta_p(2 - s^*(p,\x))-x_p.$
	Condition  2 is then satisfied by 
	$ \mb{U}^{\frac{1}{\N}}_{k+1} =$ $\smash{
		\N(\mb{X}^{\frac{1}{\N}}_{k+1} - \mb{X}^{\frac{1}{\N}}_{k})} -\bm{\nu}^{\frac{1}{\N}}(\x)$.
	As for Condition 3, we consider a selection of the differential inclusion \eqref{eq:type_mixed} whose $p^{\text{th}}$ element is equal to $\theta_p(2 - s^*(p,\x))-x_p$, and we denote it by $\bm{\nu}_1(\bm{\x})$.
	Condition 3 is then satisfied by choosing an arbitrary value of $\epsilon_0$, setting $\y=\x$, and  $\bm{\nu}=\bm{\nu}_1.$
	As for the last condition,  given $\mathbb{E}[\mb{U}^{\frac{1}{\N}}_{k+1}| \mb{X}_k^{\frac{1}{\N}}] = 0$ for $k \geq 0$,
	the 
	$\langle \mb{U}^{\frac{1}{\N}}_k\rangle_k$ is a martingale difference sequence.
	In addition, for all $k$,
	$\vert \mb{U}^{\frac{1}{\N}}_k \vert \leq \sqrt{\Sigma_{l=1}^{2\pp} (1 + \theta_l)^2}$
	and hence $\mb{U}^{\frac{1}{\N}}_k$ 
	is uniformly bounded.
	Thus, in view of \cite[Proposition 2.3]{roth2013stochastic}, Condition 4 is met.
\end{proofref}
\begin{proofref}{\Cref{{prop_transient_dynamics}}}
\tr{Given \Cref{lemGSAPS}, the collection $\mb{X}^{\frac{1}{\N}}$, for vanishing sequence $\langle \frac{1}{\N} \rangle_{\N \in \mathcal{N}}$, is a GSAP for the differential inclusion \eqref{eq:type_mixed}.
	With this and given \Cref{thm:shortTermSandholm}, the
	proof is complete.}
	\begin{proofref}{\Cref{lem_equilibriumPointofContinuous}}
		Assume $\x^*$ is an equilibrium of \eqref{eq:type_mixed}, i.e., $\smash{\x^* \in \bm{\mathcal{Q}}.}$
		It holds that $\smash{\bm 0 \in \bm{\V}(\x^*)}$ and hence $0 \in \V_p(\x^*)$ for $p \in [2\pp]$.
		Let $x^* = \sum_{i=1}^{2\pp} x^*_i.$
		\emph{Case 1.} $x^* \notin \mathcal{T}$.
		from the relation $0 \in \V_p(\x^*)$, it follows that
		 for  $\pp <p \leq \pp+\p $ (resp. $\pp + \p <p  \leq 2\pp$), if $x^* < \tau_{p-\pp}$ (resp. $x^*>\tau'_{2\pp+1-p}$),  $x^*_p = \theta_{p}$.
		In other case, for $\pp <p \leq \pp+\p $ (resp. $\pp + \p <p  \leq 2\pp$), if $x^* > \tau_{p-\pp}$ (resp. $x^*<\tau'_{2\pp+1-p}$), we have  $x^*_p = 0$.
		Let   $i = \max \{k \in [\p] \mid  \tau_k > x^* \}$ and $j = \max \{ k \in [\p']\mid  \tau'_{k} < x^* \},$ \tb{where $i=1$ (resp. $j$=0) if the former (resp. latter) set is empty.}
		It follows that 
		$x^* \in [\eta_i + \eta'_j, \eta_{i+1} + \eta'_{j+1} + \zeta_0] \label{eq:proof_equilibrium-1-0}$ and
		$\max \{\tau_{i+1}, \tau'_{j} \}  < x^* < \min \{\tau_{i}, \tau'_{j+1} \}$.
		%     \begin{subequations}
			% \label{eq:proof-equilibrium-1}
			% \begin{align}
				%    x^* &\in [\eta_i + \eta'_j, \eta_{i+1} + \eta'_{j+1} + \zeta_0] \label{eq:proof_equilibrium-1-0},\\
				%    \max \{\tau_{i+1}, \tau'_{j} \} & < x^* < \min \{\tau_{i}, \tau'_{j+1} \}
				%    \label{eq:proof_equilibrium-1-1}.
				%    \end{align}
			% \end{subequations}   
		The point $\x^*$ can then be characterized as
		\begin{equation} \scalebox{0.9
			}{$\label{eq:proof-equilibrium-clean}
				(x^*_1, x^*_2, \ldots, x^*_{\pp},\rho_1, \rho_2, \ldots, \rho_i, 0, \ldots,0, \rho'_j,\ldots,\rho'_1).$}
		\end{equation}
		The $p^\text{th}$ element of $\x^*$, $p \in [\pp]$, is associated with imitators  and  equals
		zero (resp. $\theta_p$) if at the equilibrium point $\x^*$ both relations  $u^\1(\x^*) <u^\2(\x^*)$ and $u^\1(x^*) <u^\2(x^*)$ (resp. $u^\1(\x^*) >u^\2(\x^*)$ and $u^\1(x^*) >u^\2(x^*)$) are satisfied.
		The state \eqref{eq:proof-equilibrium-clean} with $x^*_p = \theta_p$ (resp. $x^*_p = 0$) for every $p \in [\pp]$ characterizes a clean-cut (resp. imitation-free clean-cut) state.
		% Otherwise, element $p$ of $\x^*$, $p \in [\pp]$, can take any value in the interval $[0,\theta_p]$.
		The other possible case for $p^\text{th}$ element of $\x^*$, $p \in [\pp]$ occurs when $\V_p(\x^*)$, $p \in [\pp]$, is not a singleton.
		\tb{Given \Cref{lemmaAppendix}, if  $x^* \notin \mathcal{T}$, and $\V_p(\x^*)$ is not a singleton for some $p \in [\pp]$, $\x^*$ can be an equilibrium only if
		$u^\1(x^*) =u^\2(x^*)$, or, equivalently,
		$x^* \in \mathcal{C}$.}
		Putting \eqref{eq:proof-equilibrium-clean} and $x^* \in \mathcal{C}$ together, the characterization of 
		a mixed clean-cut population state is obtained.
		\emph{Case 2.} $x^* \in \mathcal{T}$.
		Assume that $x^* = \tau_i$ for $i \in [\p]$.
		A similar reasoning provided in \emph{Case 1} results in
		$
		x^* \in [\eta_{i-1} + \eta'_j, \eta_{i} + \eta'_{j} + \zeta_0],
		$
		where   $j = \max \{ k \in \{0\} \cup [\p'] \mid \tau'_{k} < x^* \}$.
		With this and given $\sum_{p=1}^{2\pp}x_p^* = \tau_i$, the equilibrium point $\x^*$ should be in the  form
		$
		\big(x^*_1, x^*_2, \ldots, x^*_{\pp},\rho_1, \ldots, \rho_{i-1}, \tau_i - (\eta_{i-1} +\eta'_j + \sum_{p=1}^{\pp} x^*_p)\\
		,0, \ldots, 0, \rho'_j,\ldots,\rho'_1\big).
		$
		If $x^* = \tau'_j$, for some $j \in [\p']$,
		a similar reasoning results in
		$
		x^* \in [\eta_{i} + \eta'_{j-1}, \eta_{i} + \eta'_{j} + \zeta_0] 
		$
		and
		the  characterization
        $$
		\big(x^*_1, x^*_2, \ldots, x^*_{\pp},\rho_1, \ldots, \rho_{i}, 0,
		\ldots, 0, \tau'_j - (\eta'_{j-1} +\eta_i + \sum_{p=1}^{\pp} x^*_p)
		,\rho'_{j-1},\ldots,\rho'_1\big)
		$$  for
		point $\x^*$ .
		% Hence, the population states defined in \Cref{def_population-equilibrium-point}  include all possible forms that an equilibrium point $\x^*$  can take.
		Similarly to the argument in \emph{Case 1}, 
		given \Cref{lemmaAppendix}, it is easy to show that if $x^* \notin \mathcal{C}$  and $\x^* \in \bm{\mathcal{Q}}$, then $\V_p(\x^*)$, $p \in [\pp]$, is a singleton.
		Hence, the $p^\text{th}$ element of $\x^*$, $p \in [\pp]$, is   zero (resp. $\theta_p$) if at the equilibrium point $\x^*$,   $u^\1(\x^*) <u^\2(\x^*)$  (resp. $u^\1(\x^*) >u^\2(\x^*)$) is satisfied.
		These conditions characterize a non-mixed nonconformist-driven for $x^* = \tau_i$, $i \in [\p]$ and non-mixed conformist-driven for $x^* = \tau'_j$, $j \in [\p'].$
		% The following cases may happen, depending on the value of $x^*$:
		%  Assume that   $x^*$ equals $\tau_p$ for some $p \in [\pp].$
		%   Two cases might happen:
		%   \emph{Case 1)} $u^\1(\tau_p) \neq  u^\2(\tau_p)$. 
		%   If $u^\1(\tau_p) > u^\2(\tau_p)$ (resp. $u^\1(\tau_p) < u^\2(\tau_p)$), then for some $k \in [\pp]$, $p \neq k$, $u^\1(\tau_p) = u^\1_k(\tau_p)$ and $u^\1_k(\tau_p) > u^\2_k(\tau_p).$ (resp. $u^\2(\tau_p) = u^\2_k(\tau_p)$ and $u^\2_k(\tau_p) > u^\1_k(\tau_p)$).
		%   Since $\x^*$ is an equilibrium, it holds that $x_{\pp+q} = \theta_{\pp+q}$ (resp. $x_{\pp+q} = 0$), and accordingly
		%    the maximum active $\1$-utility 
		%   is also greater than that of $\2$-utility (resp. the maximum active $\2$-utility  is also greater than that of $\1$-utility).
		% This implies that  $x^*_p$ equals $\theta_p$ (resp. $0$) for $p \in [\pp]$.	
		%  If $\tau_p = \tau_{i}$ and $i \in [\p]$, 
		%  then $\V_{\pp + i}(\x^*) = [-x^*_{\pp + i}, \theta_{\pp + i}-x^*_{\pp+i}]$ and    
		%  $0 \in \V_{\pp + i}(\x^*)$,
		%  the condition \eqref{eq:proof_equilibrium-1-0}  takes the form $x^* \in [\eta_{i-1} + \eta'_j + \zeta_0, \eta_i + \eta'_{j} + \zeta_0]$ (resp. $x^* \in [\eta_{i-1} + \eta_{j} , \eta_i + \eta'_{j} ]$); otherwise
		%    if $\tau_p = \tau'_j$ and $j \in [\p']$,  the condition \eqref{eq:proof_equilibrium-1-0}  gets the form $x^* \in [\eta_{i} + \eta'_{j-1} + \zeta_0, \eta_i + \eta'_{j} + \zeta_0]$ (resp. $x^* \in [\eta_{i} + \eta_{j-1} , \eta_i + \eta'_{j}] $).
		%   \emph{Case 2)} 
		Now assume that  $x^* \in \mathcal{T} \cap \mathcal{C}$.
		Let $x^* = \tau_i$, for $i \in [\p]$; the case $x^* = \tau'_j$, for $j \in [\p']$, can be handled similarly.
		The element $p$ of $\x^*$, $p \in ([\pp] \cup \{i+\pp\})$, can take any value in the interval $[0,\theta_p]$ as long as $\sum_{p=1}^\pp x_p^* + x_{i + \pp}^* = \tau_i - (\eta_{i-1} + \eta'_j)$.
		\tb{The remaining elements of  $\x^*$ can be characterized similar to the cases studied earlier.}
		This characterizes a mixed nonconformist-driven state.
		Hence, every equilibrium point $\x^*$ of \eqref{eq:type_mixed} is one of the states defined in \Cref{def_population-equilibrium-point}.
	\end{proofref}
	\begin{figure*}
		\begin{gather} \label{eq:populationProportion}
			\dot{\tilde{x}} \in \tilde{{\mathcal{{X}}}}(\tilde{x}, \x), \quad \tilde{x}(0) = \bm 1^\top \x(0), \\
			\scalebox{0.75}{$
				{\tilde{\mathcal{{X}}}}(\tilde{x}, \x) = \begin{cases}
					[\eta_{i-1} + \eta'_{j} + \zeta_0-\tilde{x}, \eta_{i} + \eta'_{j} + \zeta_0-\tilde{x}],  &\text{if }  \exists i (\tilde{x} = \tau_{i}),
					u^\1 (\tilde{x}) > u^\2(\tilde{x}), \text{ and } u^\1 (\x) > u^\2(\x),
					\\
					[\eta_{i} + \eta'_{j-1} + \zeta_0-\tilde{x}, \eta_{i} + \eta'_{j} + \zeta_0-\tilde{x}],  &\text{if }  \exists j (\tilde{x} = \tau_{j}'),
					u^\1 (\tilde{x}) > u^\2(\tilde{x}) , \text{ and } u^\1 (\x) > u^\2(\x),\\
					[\eta_{i-1} + \eta'_{j}-\tilde{x}, \eta_{i} + \eta'_{j}-\tilde{x}],  &\text{if }  \exists i (\tilde{x} = \tau_{i}),
					u^\1 (\tilde{x}) < u^\2(\tilde{x}) , \text{ and } u^\1 (\x) < u^\2(\x),
					\\
					[\eta_{i} + \eta'_{j-1}-\tilde{x}, \eta_{i} + \eta'_{j}-\tilde{x}],  &\text{if }  \exists j (\tilde{x} = \tau_{j}'),
					u^\1 (\tilde{x}) < u^\2(\tilde{x}), \text{ and } u^\1 (\x) < u^\2(\x),\\
					[\eta_{i-1} + \eta'_{j}-\tilde{x}, \eta_{i} + \eta'_{j}+ \zeta_0-\tilde{x}],  &\text{if }  \exists i (\tilde{x} = \tau_{i}) \text{ and }
					\Big( u^\1 (\tilde{x}) = u^\2(\tilde{x}) 
					\text{ or } \\ & \quad
					\operatorname{sgn}(u^\1(\bm 1^\top \x) -u^\2(\bm 1^\top \x)) \neq \operatorname{sgn}(u^\1(\x) -u^\2(\x)) \Big),
					\\
					[\eta_{i} + \eta'_{j-1}-\tilde{x}, \eta_{i} + \eta'_{j}+\zeta_0-\tilde{x}],  &\text{if }  \exists j (\tilde{x} = \tau_{j}') \text{ and }
					\Big( u^\1 (\tilde{x}) = u^\2(\tilde{x}) 
					\text{ or } \\ & \quad
					\operatorname{sgn}(u^\1(\bm 1^\top \x) -u^\2(\bm 1^\top \x)) \neq \operatorname{sgn}(u^\1(\x) -u^\2(\x)) \Big),\\
					[\eta_{i} + \eta'_{j} -\tilde{x}, \eta_{i} + \eta'_{j} + \zeta_0 -\tilde{x}],      &\text{if }  \exists i,j (\tilde{x} \in   (\max \{\tau'_{j}, \tau_{i+1}\}, \{\tau'_{j+1}, \tau_{i} \})) 
					\\ 
					& \text{and }\Big( u^\1 (\tilde{x}) = u^\2(\tilde{x})\text{ or } \operatorname{sgn}(u^\1(\bm 1^\top \x) -u^\2(\bm 1^\top \x)) \neq \operatorname{sgn}(u^\1(\x) -u^\2(\x)) \Big),
					\\   
					\{\eta_{i} + \eta'_{j}+\zeta_0 -\tilde{x}\}, &\text{if }  \exists i,j (\tilde{x} \in   (\max \{\tau'_{j}, \tau_{i+1}\}, \{\tau'_{j+1}, \tau_{i} \})),
					u^\1 (\tilde{x}) \!>\! u^\2(\tilde{x}), \text{ and }  u^\1 \!(\x) \!>\! u^\2(\x),
					\\
					\{\eta_{i} + \eta'_{j} -\tilde{x}\},    &\text{if }  \exists i,j (\tilde{x} \in   (\max \{\tau'_{j}, \tau_{i+1}\}, \{\tau'_{j+1}, \tau_{i} \})),
					u^\1 (\tilde{x})\! < \!u^\2(\tilde{x}), \text{ and }  u^\1 \!(\x) \!<
					\!u^\2(\x).
				\end{cases}
				$ } \nonumber
		\end{gather}
	\end{figure*}
	\begin{proofref}{\Cref{lem_abstractEquilibrium}}
		The lemma can be proved by applying the definition of an equilibrium point of a differential inclusion and following the steps in the proof of \cite[Lemma 2]{aghaeeyan2023discrete}, and omitted due to space limitation.
	\end{proofref}	
	\begin{proofref}{\Cref{lem_instability_of_mixed_clean_cut}}
		We prove the instability of the mixed clean-cut equilibrium points.
		The instability of the conformist-driven equilibrium points and the stability of the remaining equilibrium points can be shown using arguments similar to in the proof of Theorem 2 in \cite{aghaeeyan2023discrete}.
		Let $a_p$, $b_p$, $c_p$, and $d_p$ denote the entries of the payoff matrix associated with type $p$.
		It follows that $u^\1_{i(q)}(0)$ (resp. $u^\2_{i(q)}(0)$) is equal to $b_{i(q)}$ (resp. $d_{i(q)}$).
		Let $i = i(q)$ (resp. $j = j(q)$).
		We prove that the inequality $b_i < d_j$ implies the instability of the mixed clean-cut equilibrium point $q$.
		From $u^\1(q) = u^\1_i(q)$, $u^\2(q) = u^\2_j(q)$, and $u^\1(q) = u^\2(q)$, we conclude $(a_i - b_i)q + b_i = (c_j - d_j)q + d_j$.
		The relations $b_i < d_j$ and $q > 0$ imply $(-a_i - d_j +b_i + c_j)<0$. This results in
		$q = (b_i - d_j)/(-a_i - d_j +b_i + c_j)$. 
		We first show that the statement
		\begin{equation} \label{equation_instability_of_mixed_clean_cut}
			\exists \epsilon >0 \forall x \in (q- \epsilon, q ), \bigl( u^\1(x) < u^\2(x)\bigr)
		\end{equation} 
		implies the inequality  $b_i < d_j$.
		Under \eqref{equation_instability_of_mixed_clean_cut}, 
		$u^\1_i(q - \epsilon) < u^\2_j(q - \epsilon)$, for an small enough positive $\epsilon.$ 
		It follows that
		\begin{equation} \label{equation_instability_mixed_clean_cut2}
			(b_i - d_j) < (c_j - d_j + b_i - a_i)(q - \epsilon).
		\end{equation}
		If $b_i - d_j >0$, due to the fact that  $b_i - d_j$ and $(c_j - d_j + b_i - a_i)$ have the same sign, the equation \eqref{equation_instability_mixed_clean_cut2} implies
		$q < q - \epsilon$, which contradicts the positivity of $\epsilon.$
		If $b_i = d_j$, either $q$ is zero or the term $(c_j - d_j + b_i - a_i)$ is zero. 
		It can be easily shown that the origin is not an equilibrium, which implies that $q$ cannot be zero.
		In addition, according to \Cref{ass:ass3}, if $b_i = d_j$, then $i = j$. This implies that the point $q$ is equal to some threshold. 
		However, given the definition of a mixed clean-cut equilibrium point, $q$ does not equal any threshold.
		Therefore,  $b_i - d_j <0$.
		Similarly, it can be shown that the relation $b_i < d_j$ implies
		\eqref{equation_instability_of_mixed_clean_cut}.
		Given \eqref{equation_instability_of_mixed_clean_cut}, the rest of the proof can be shown using the definition of stability and follows from the proof of Theorem 2 in \cite{aghaeeyan2023discrete}.
		%  On the other hand, 
		% note that whenever $b_i <  d_j$,
		% $i$ (resp. $j$) is the type with the greatest $\1$-utility (resp. $\2$-utility) at $x = q$, and
		%  $b_i$ (resp. $d_j$) is $u^\1_i(0)$ (resp.  $u^\2_j(0)$).
	\end{proofref}
		\begin{proofref}{\Cref{basin-of-attraction-of-abstract-dynamics}}
		Given that the abstract dynamic is one-dimensional, along with the definition of basin of attraction, the instability of equilibrium points $q^*_{k,k+1}$, for $k \in [\mathtt{Q}-1]$, and the asymptotic stability of equilibrium points $q^*_k$, for $k\in [\mathtt{Q}]$, the proof is immediate.
	\end{proofref}
	\begin{proofref}{\Cref{lemma_correspondece_abstract_population_eq}}
		We prove for a nonconformist-driven equilibrium point.
		Other types of equilibria can be handled similarly.
		We show for the case where $q \in \mathcal{C}$; other cases can be handled similarly.
		Let $q = \tau_i$, for some $i \in [\p]$.
		It holds that $q \in (\eta_{i-1} + \eta_j', \eta_i + \eta_j' + \zeta_0)$.
		Any point $\bm q$ in the set $\bm{\mathcal{R}}^n_{i,j}$ satisfies the first part of the lemma for the abstract equilibrium point $q$. This is straightforward because the sum of all $2\pp$ elements in each member of $\bm{\mathcal{R}}^n_{i,j}$ equals $\tau_i.$
		Now, we move on to the second part. 
		Assume, on the contrary, there exists an equilibrium point $\bm q \in \bm{\mathcal{Q}}$ such that 
		$\sum_{p=1}^{2\pp}q_p$ is not an abstract equilibrium point.
		By \Cref{lem_equilibriumPointofContinuous}, any equilibrium point of the continuous-time dynamics is either of the seven categories defined in \Cref{def_population-equilibrium-point}.
		If $\bm q$ is a clean-cut (resp. imitation-free clean-cut), then $\sum_{p=1}^{2\pp}q_p$ equals $q= \eta'_j + \eta_i + \zeta_0$ (resp. $q= \eta'_j + \eta_i$) satisfying 
		$\max \{\tau_{i+1}, \tau'_{j} \}< q < \min \{\tau_{i}, \tau'_{j+1} \}$ and $u^\1(\eta_i + \eta_j'+ \zeta_0) > u^\2(\eta_i + \eta_j'+ \zeta_0)$ (resp.   $\max \{\tau_{i+1}, \tau'_{j} \}< {\eta}'_{j}+ {\eta}_{i}  < \min \{\tau_{i}, \tau'_{j+1} \}$ and $u^\1(\eta_i + \eta_j') < u^\2(\eta_i + \eta_j')$).
		This is, however, the characterization of a clean-cut (resp. imitation-free clean-cut) equilibrium point of the abstract dynamics.
		Hence, the point $\bm q$ cannot be a clean-cut (resp. imitation-free clean-cut) one.
		A similar line of reasoning for other types of continuous-time population dynamics leads to a contradiction.
		This completes the proof.
	\end{proofref}
	\subsection*{Supporting lemmas for the proof of \Cref{lem:attraction-population-single-equilibrium}}
We provide two additional lemmas which facilitate the proof of \Cref{lem:attraction-population-single-equilibrium}.  
	%Azadeh May 18th Given \eqref{eq:populationProportion} and \eqref{eq:abstract_het}, it can be easily shown that, upon starting from the same initial condition, as long as the relation
	%    \begin{equation} \label{equal_sgn}
		%       \scalebox{0.87}{$
			% \operatorname{sgn}\big(u^\1(\x(t))\!- \!u^\2(\x(t))\!\big) \!=\! \operatorname{sgn}\!\big(u^\1(\displaystyle\sum_{p=1}^{2\pp}x_p(t))\!-\! u^\2(\displaystyle\sum_{p=1}^{2\pp}x_p(t))\!\big)$}
		% \end{equation}
	%Azadeh May 18th  is true, the active cases in \eqref{eq:populationProportion} and \eqref{eq:abstract_het} are the same.
	To do so, we define the set-valued map $\mathcal{S}(\x)$ which maps each $\x$ in  $\bm{\mathcal{X}}_{ss}- \bm{\mathcal{Q}}$ to a subset $\mathcal{S}(\x)$ of $\{\pp+1,\pp+2,\ldots, 2\pp\}$ such that any $p \in \mathcal{S}(\x)$ satisfies  
	% \begin{enumerate}
		% \item
		$\tau_{p-\pp} \in \mathcal{C}$,
		% \item     $\smash{\tau_{\pp-p} \in (\min\{\sum_{p=1}^{2\pp}x_p), q \}, \max \{\sum_{p=1}^{2\pp}x_p), q \})}$, where $q \in \mathcal{Q}$ and $\mathcal{A}(q) \ni \sum_{p=1}^{2\pp}x_p$;
		% \item 
		$x_{p-\pp} \in \{0,\theta_{p-\pp}\}$, and
		% \item     
		$x_{p} \in \{0,\theta_{p}\}$, where 
		for $\p < j < \pp$, we define ${\tau}_j = {\tau}'_{\pp +1-j}$.
		Indeed, each element of the set $\mathcal{S}(\x)$ is the index of a best-responding type whose threshold belongs to $\mathcal{C}$ and whose associated components (including the corresponding imitative component) in $\x$ are at the boundary.
		These types may cause the maximum utility and \tb{maximum} active utility to diverge at population states $\x$ where $\bm{1}^\top \x$ equals the threshold associated with them.
		% \end{enumerate}
	\begin{lemmaAppendix} \label{lem_set_S}
		Consider the continuous-time population dynamics \eqref{eq:type_mixed}.
		Assume  the set $\mathcal{S}(\x(0))$ is empty, and 
		\begin{equation} \label{equal_sgn}
			\scalebox{1}{$
				\operatorname{sgn}\big(u^\1(\x(t))\!- \!u^\2(\x(t))\!\big) \!=\! \operatorname{sgn}\!\big(u^\1(\displaystyle\sum_{p=1}^{2\pp}x_p(t))\!-\! u^\2(\displaystyle\sum_{p=1}^{2\pp}x_p(t))\!\big)$}
		\end{equation}
		holds for $t = 0$.
		Then, under Assumptions \ref{ass:ass1}-\ref{ass:ass3},
		the equality \eqref{equal_sgn} holds for all $t\geq 0.$
	\end{lemmaAppendix}
	The proof is omitted due to space limitations.
	\begin{lemmaAppendix} \label{lem_convergence_of_population_proportion}
		Assume the initial condition of the continuous-time population dynamics satisfies 
		$\bm 1^\top \x(0) \in \mathcal{A}(q^*_k)$ for some $k \in [\mathtt{Q}]$.
		Then,
		under Assumptions \ref{ass:ass1}-\ref{ass:ass3},
		the population proportion of $\1$-players converges to $q^*_k.$
	\end{lemmaAppendix}
	
	To prove this, we make a connection between the evolution of the 
	abstract dynamics \eqref{eq:abstract_het} and that of the population proportion of $\1$-players \eqref{eq:populationProportion} when both dynamics share the same initial condition $\bm 1^\top \x_0 \in \mathcal{A}(q^*_k)$ for some $k \in [\mathtt{Q}]$.
	Whether the trajectories of the two scalar dynamics \eqref{eq:abstract_het} and \eqref{eq:populationProportion} coincide for all $t>0$ depends on whether the equality \eqref{equal_sgn} holds at $t=0,$ and if the initial condition lies on the boundary.
	Even if they diverge at some point, we show that the active cases for  \eqref{eq:abstract_het} and \eqref{eq:populationProportion}, at some finite time, will be and remain the same.
	At that point, we reset the abstract state to match the current population proportion of $\1$-players. 
	Since the abstract dynamics is only an auxiliary tool, this reset \tb{is valid and} helps us show that the population proportion converges to $q^*_k$ as long as $\bm 1^\top \x_0 \in \mathcal{A}(q^*_k)$.
	Note that the population dynamics follow some selection of the differential inclusion in \eqref{eq:populationProportion}. 
	% So, when the trajectories align, we implicitly treat the abstract state as tracking the evolution of the population proportion.
	Accordingly, in the following proof, when the sets of trajectories of the two systems align, we implicitly study the selection of the abstract dynamics that matches the one governing the population proportion of $\1$-players.
	In the following proof, the notation $t^-$  (resp. $t^+$) indicates the left-hand (resp. right-hand) limit point of the real line at $t$, or, intuitively, the time instant immediately before (resp. after) time $t$. 
	
%	We prove the case where the point $q^*_k$  is a nonconformist-driven equilibrium point of the abstract dynamics and, hence
%	equals some nonconformist threshold $\tau_i$, $i \in [\p].$
%	Other cases can be handled similarly.
	Consider an initial condition $\x_0$ satisfying $\bm 1^\top \x_0 \in \mathcal{A}(q^*_k)$. 
	We further assume that $k \notin \{1,\mathtt{Q}\}$;
	a similar reasoning can be applied to the case  $k \in \{1,\mathtt{Q}\}$.
	By \Cref{basin-of-attraction-of-abstract-dynamics}, 
	$\mathcal{A}(q^*_k) =(q^*_{k-1,k}, q^*_{k,k+1}).$
	Now, two cases may happen based on the initial condition $\x_0$: \emph{Case 1}, where
	\eqref{equal_sgn}
	holds for $t=0$, and 
	\emph{Case 2} otherwise.
	\emph{Case 1}.
	Here, if the abstract dynamics \eqref{eq:abstract_het} start from the same initial condition 
	as the dynamics \eqref{eq:populationProportion} start, i.e. $\hat{x}(0) = \bm 1^\top \x_0$,  the active cases in both systems are the same.
	Part 1) If $\x_0 \notin \partial \bm{\mathcal{X}}_{ss}$,
	for any finite time $t>0$, we have $\x(t) \notin \partial \bm{\mathcal{X}}_{ss}$ and hence
	 \eqref{equal_sgn} holds.
	This yields the equivalency of the sets of solution trajectories of both systems \eqref{eq:abstract_het} and \eqref{eq:populationProportion}.
	On the other hand, given \Cref{basin-of-attraction-of-abstract-dynamics}, if the initial condition satisfies $\hat{x}(0) \in \mathcal{A}(q^*_k)$, the abstract dynamics approach the point $q^*_k$.
	This results in the convergence of the population proportion of $\1$-players to point $q^*_k.$
	Part 2) \tb{$\x_0 \in \partial \bm{\mathcal{X}}_{ss}$ and} $\mathcal{S}(\x(0)) = \varnothing$.
	Given \Cref{lem_set_S},
	 \eqref{equal_sgn} holds for all time and, accordingly, the reasoning provided in Part 1 is applicable.
	Part 3) \tb{$\x_0 \in \partial \bm{\mathcal{X}}_{ss}$ and}  $\mathcal{S}(\x(0)) \neq  \varnothing$. 
	%  We disregard those elements $p$ of the set $\mathcal{S}$ satisfying
	% $\smash{\tau_{\pp-p} \notin (\min\{\sum_{p=1}^{2\pp}x_p), q^*_k \}, \max \{\sum_{p=1}^{2\pp}x_p), q^*_k\})}$,
	% as the value of $\sum_{p=1}^{2\pp}x_p(t)$ does not reach  the threshold $\tau_{\pp-p}$ and, consequently, the value of $u^\1_{\pp-p}(x)$ (resp. $u^\2_{\pp-p}(x)$) remains less than that of the maximum $\1$-utility (resp. $\2$-utility) if $x_p(0) =0$ (resp. $x_p(0) =\theta_p$) for all time.
	% % the active case for $\V_{\pp + p}(\x)$ in \eqref{eq:type_mixed} remains the same for all time.
	% We put the remaining elements of the set $\mathcal{S}$, i.e., those $p \in \mathcal{S}$ satisfying 
	%  $\smash{\tau_{\pp-p} \in (\min\{\sum_{p=1}^{2\pp}x_p), q^*_k \}, \max \{\sum_{p=1}^{2\pp}x_p), q^*_k\})}$, in a separate set denoted by $\mathcal{S}_1$
	Let 
	% \begin{equation*}
		%$\mathcal{I} \!= \! \left\{\!t\! \mid \! \sum_{k=1}%^{2\pp}\!x_k(t) \!= \!\tau_{p - \pp} \text{ for some  $p\! \in \!\mathcal{S}(\x(0))\!$ and  $t\!<\!%%\infty$\!}
		%\right\},$
        $\mathcal{I} \!= \! \left\{\!t>0\! \mid \! \exists M>0,p\! \in \!\mathcal{S}(\x(0)) \big( \sum_{k=1}^{2\pp}\!x_k(t) \!= \!\tau_{p - \pp} ,  t\!<\!M \big)
		\right\}$
		% \end{equation*}
	where the evolution of $x_k$, $k \in [2\pp]$, is governed by \eqref{eq:type_mixed}. 
	If $\mathcal{I} = \varnothing$, then, for all time, the inactive utilities associated with at-the-boundary elements of the population state are strictly less than the corresponding active utilities. 
	Accordingly, the reasoning provided in Part 1 is valid. 
	Now, let
	$\mathcal{I} \neq \varnothing.$
	Let $t_1 = \min \{t \mid t \in \mathcal{I} \}$ and denote by $p_1$ the 
	member of the set $\mathcal{S}(\x(0))$ associated with $t_1$,
	that is, $ \sum_{k=1}^{2\pp}x_k(t_1) =\tau_{p_1 - \pp}.$
    We assume $\bm 1^\top \x(0) < q^*_k$ and analyze the case where $\tau_{p_1-\pp}$ is the threshold of a conformist agent, i.e.,
	$u^\1(\sum_{k=1}^{2\pp} x_k(t_1-\epsilon))<u^\2(\sum_{k=1}^{2\pp} x_k(t_1-\epsilon))$, for a small enough positive $\epsilon.$
	Other cases can be handled similarly.
	Following the argument provided above, it is immediate that for $t < t_1$, the set of trajectories of the abstract dynamics and population proportion of $\1$-players are equal.
	If $x_{p_1}(t_1^-) \in (0, \theta_{p_1})$ or $x_{{p_1}-\pp}(t_1^-) \in (0, \theta_{{p_1}-\pp})$,  both $u^\1_{{p_1-\pp}}$ and $u^\2_{{p_1 -\pp}}$ are active and \tb{included in computing} the maximum active utilities.
	It follows that the evolution of abstract state \eqref{eq:abstract_het} and the population proportion of $\1$-players  \eqref{eq:populationProportion} remain the same for $t \in (t_1 - \epsilon_1, t_1 + \epsilon_1)$ for some small enough value $\epsilon_1$.
	Now, consider the case where both $x_{p_1}(t_1^-)$ and $x_{{p_1}-\pp}(t_1^-)$ are zero.
	At the state $\x(t_1)$,
	zero is an element of the set-valued maps $\V_{p_1}(\x(t_1))$ and $\V_{ {p_1} - \pp}(\x(t_1))$.
	For all possible values of  $\dot{x}_{p_1}(t_1) \in \V_{p_1}(\x(t_1))$ and  $\dot{x}_{{p_1}-\pp}(t_1) \in \V_{ {p_1} - \pp}(\x(t_1))$ but zero, 
	the evolution of the dynamics \eqref{eq:abstract_het} and \eqref{eq:populationProportion} remain the same, and the arguments provided for previous cases hold valid.
	Otherwise, when  both $\dot{x}_{p_1}$ and $\dot{x}_{p_1-\pp}$ are zero at $t = t_1$,
	the active cases for the two dynamics
	\eqref{eq:abstract_het} and \eqref{eq:populationProportion}
	differ.
    The point $\tau_{p_1-\pp} \neq q^*_k$ is not an equilibrium  of \eqref{eq:abstract_het}--point $q^*_k$ is the only equilibrium 
	on the interval $\mathcal{A}(q^*_k)$)--and, in turn, $0 \notin \mathcal{X}(\tau_{p_1-\pp})$.
	In addition, 
	$\bm 1^\top \x(0) < q^*_k$ results in
	$\min \big(y \in \mathcal{X}(\tau_{p_1-\pp})\big) > 0$ and hence $\min \big(y \in \mathcal{\tilde{X}}(\tau_{p_1-\pp})\big) > 0$.
	It follows that at time $t = t_1^+$ both abstract state and $\sum_{k=1}^{2\pp}x_k$ are greater than $\tau_{{p_1-\pp}}.$
	Accordingly, 
	the dynamics for \eqref{eq:populationProportion} at $t = t_1^+$ will be
	$[\eta_{i} + \eta'_{j} -\sum_{k=1}^{2\pp}x_k, \eta_{i} + \eta'_{j} + \zeta_0 -\sum_{k=1}^{2\pp}x_k]$, whereas for
	\eqref{eq:abstract_het} we have
	$\{\eta_{i} + \eta'_{j} + \zeta_0 -\hat{x}\}$ with $j = 2\pp+1-p_1$ and $i = \max \{k \leq \p \mid \tau_k > \tau'_j\}$.
	In view of (\ref{eq:type_mixed}d),  $\sum_{p=1}^{2\pp}x_p(t_1^+) >\tau_{p_{1}-\pp}$, we have $\dot{x}_{ {p_1}}(t_1^+) = \theta_{ {p_1} } - x_{ {p_1}}(t_1^+)$ which yields $x_{ {p_1}}(t_1^+) >0$.
	This implies that \eqref{equal_sgn} for $t = t_1 + \varepsilon_1$, where $\varepsilon_1$ is small enough, holds.
	Now, we reset the abstract dynamics by setting the abstract state at time point $t_1 + \varepsilon_1$  equal to $\sum_{k=1}^{2\pp} x_k(t_1 + \varepsilon_1)$.
	Then, the evolution of the abstract dynamics and the population proportion of $\1$-players are equivalent (\eqref{eq:abstract_het} and \eqref{eq:populationProportion}).
	It is straightforward to show that 
	$\sum_{k=1}^{2\pp} x_k(t_1 + \varepsilon_1)\in \mathcal{A}(q^*_k)$
	and, accordingly, both move towards $q^*$ up until the
	population proportion of $\1$-players
	reaches the second smallest member of the set $\mathcal{I}$.
	We follow the preceding reasoning for the second smallest member of the set $\mathcal{I}$ and for other members of the set $\mathcal{I}$.
	Let $t^*$ be the greatest member of the set  $\mathcal{I}$.
	By resetting the abstract dynamics and setting its state equal to $\sum_{k=1}^{2\pp}x_k(t^*+\varepsilon)$, for some small enough positive value $\varepsilon,$
	the evolution of these two dynamics remains the same for all finite time, and
	the rest of the proof follows that of the case $\x(0) \notin \bm{\mathcal{X}}_{ss}.$
	Now we study \emph{Case 2},
	where \eqref{equal_sgn} does not hold at $t = 0$.
	Without loss of generality, let $r\in [\pp]$ and $q \in [\pp]$ satisfy
	$u^\1(\sum_{k=1}^{2\pp}x_k(0)) = u^\1_q(\sum_{k=1}^{2\pp}x_k(0))$ and $u^\2(\sum_{k=1}^{2\pp}x_k(0)) = u^\2_r(\sum_{k=1}^{2\pp}x_k(0))$, respectively.
	First, assume $\sum_{k=1}^{2\pp}x_k(0) \notin \mathcal{T}$.
	In this case, while
	$u^\1_q(\sum_{k=1}^{2\pp}x_k(0)) > u^\2_q(\sum_{k=1}^{2\pp}x_k(0))$ and 
	$u^\2_r(\sum_{k=1}^{2\pp}x_k(0)) > u^\1_r(\sum_{k=1}^{2\pp}x_k(0))$ hold,
	either $x_{\pp + q}(0) = 0$ or 
	$x_{\pp + r}(0) = \theta_{\pp + r}$ is true.
	In the former case, 
	the set-valued map $\V_{\pp+q}(\x(0))$ equals $\{\theta_{\pp+q}-x_{\pp+q} \}$ and, in turn,
	$\dot{x}_{\pp + q}(0) >0$. 
	Similarly, in the latter case, we have
	$\dot{x}_{\pp + r}(0) <0$.
	Then, for some small enough positive value $\epsilon$, the elements $x_{\pp + q}(\epsilon)$ and $x_{\pp + r}(\epsilon)$ will not be \tb{equal to $0$ and $\theta_{\pp + r}$, respectively}, and the equation \eqref{equal_sgn} holds at $t=\epsilon.$
	Hence, by resetting the abstract dynamics at $t=\epsilon$ and setting the abstract state equal 
	to  $\sum_{p=1}^{2\pp}x_p(\epsilon)$, the evolution of the abstract state and the population proportion of $\1$-player will be the same for \tb{some} $t > \epsilon$. 
	On the other hand, given that
	\emph{(i)} $\epsilon$ can be arbitrarily small,  \emph{(ii)} the solution to  \eqref{eq:populationProportion} is continuous, and \emph{(iii)}
	$\bm 1^\top \x_0 \in (q^*_{k-1,k}, q^*_{k,k+1});$ \tb{it holds that}
	$\bm 1^\top \x(\epsilon) \in (q^*_{k-1,k}, q^*_{k,k+1}).$
	Thus, the rest of the proof is similar to \emph{Case 1}.
	Now assume $\sum_{k=1}^{2\pp}x_k(0) \in \mathcal{T}$ which implies 
	$\sum_{k=1}^{2\pp}x_k(0) = \tau_p$ for some $p \in [\pp].$
	If $p \notin \{r,q\}$, the arguments provided for $\sum_{k=1}^{2\pp}x_k(0) \notin \mathcal{T}$ is valid.
	Otherwise, the relation $u^\1_p(\tau_p) = u^\2_p(\tau_p)$ implies $q =p$, $r = p$, and
	$\tau_p \in \mathcal{C}$.
	Since the equation \eqref{equal_sgn} does not hold,
	either $x_{p+\pp} =0$ or $x_{p+\pp} =\theta_{p+\pp}$ is true.
	We focus on the case where $x_{p+\pp} =0$, $\bm 1^\top \x(0) < q^*_k$, and  
	$\tau_p$ is the threshold of a conformist, i.e., $p > \pp + \p$; other cases can be handled similarly.
	Accordingly, in view of \eqref{eq:abstract_het} and \eqref{eq:populationProportion}, and the fact that $\tau_p$ is not an equilibrium point for the abstract dynamics, it is easy to conclude that the values in the sets $\mathcal{X}$ and $\tilde{\mathcal{X}}$ at $\tau_p$ are bounded away from zero and are positive.
	This implies that both  abstract state and  population proportion of $\1$-players are greater than $\tau_p$ at $t= 0^+$.
	The active case for \eqref{eq:populationProportion} at $t = 0^+$ will be
	$[\eta_{i} + \eta'_{j} -\sum_{k=1}^{2\pp}x_k, \eta_{i} + \eta'_{j} + \zeta_0 -\sum_{k=1}^{2\pp}x_k]$, whereas the active case for
	\eqref{eq:abstract_het} is
	$\{\eta_{i} + \eta'_{j} + \zeta_0 -\hat{x}\}$ with 
	with $j = \pp+1-p$ and $i = \max \{k \leq \p \mid \tau_k > \tau'_j\}$.
	In view of (\ref{eq:type_mixed}d), $\dot{x}_{ {p} - \pp}(t_1^+) = \theta_{ {p} - \pp } - x_{ {p} - \pp}(t_1^+)$ yielding $x_{ {p} - \pp}(t_1^+) >0$.
	This implies that the relation \eqref{equal_sgn} for $t =  \varepsilon$, where $\varepsilon$ is a small enough positive value, holds.
	Now, we reset the abstract dynamics and set the abstract state at time point $\varepsilon_1$  equal to $\sum_{k=1}^{2\pp} x_k( \varepsilon_1)$, where $\varepsilon_1$ is a small enough positive value.
	Then, the evolution of the abstract dynamics and the population proportion of $\1$-players are equivalent (\eqref{eq:abstract_het} and \eqref{eq:populationProportion}).
	If $\mathcal{S}(\x(\varepsilon_1))$ is empty, then the set of trajectories of the abstract dynamics and that of \eqref{eq:populationProportion} are the same and, accordingly, the population proportion of $\1$-players converges to the point $q^*_k$.
	If $\mathcal{S}(\x(\varepsilon_1))$ is not empty, then a similar approach adopted in Case 1-Part 3 can be used to handle the rest of the proof.
	This completes the proof.
\end{proofref}
\begin{proofref}{\Cref{lem:attraction-population-single-equilibrium}}
	%The proof is based on the analysis of the evolution of the population state $\x$ when the population proportion of strategy-$\1$ players is in the vicinity of  $q^*_k$.
	% In the proof, we refer to \Cref{eq:populationProportion} which captures the evolution of $\sum_{p=1}^{2\pp}x_p(t)$.
	Item 1).
	Let $q^*_k = \bm 1^\top \bm q$ denote the corresponding abstract equilibrium point to $\bm q \in \bm{\mathcal{Q}},$
	the existence of which is guaranteed by \Cref{lemma_correspondece_abstract_population_eq}.
	We prove the case where the point $\bm q$ is a non-mixed nonconformist-driven equilibrium point of \eqref{eq:type_mixed}; other cases can be handled similarly.
	The value of $q^*_k$ is then equal to some nonconformist threshold $\tau_i$, $i \in [\p].$
	Consider an initial condition $\x_0$ satisfying
	$\bm 1^\top \x_0 \in \mathcal{A}(q^*_k)$. 
	We further assume that $k \notin \{1,\mathtt{Q}\}$;
	a similar reasoning can be applied to the case  $k \in \{1,\mathtt{Q}\}$.
	By \Cref{lem_convergence_of_population_proportion}, it can be easily shown that the population proportion of $\1$-player approaches an $\epsilon$-neighborhood of $q^*_k$ in finite time and never leaves it.
	We take $\epsilon$ to be small enough such that in the $\epsilon$-neighborhood of the point $q^*_k$, except for the $(\pp +i)^{\text{th}}$ element of vector $\x$, which corresponds to nonconformist type-$i$, the active cases in \eqref{eq:type_mixed} do not change, i.e.,
	$\big((q^*_k-\epsilon, q^*_k+\epsilon) - \{q^*_k\}\big) \cap \big(\mathcal{C} \cup \mathcal{T} \big) = \varnothing.$ 
	Then, the set-valued map $\V_p$ for $p \in (\{\pp+1,\ldots,2\pp\}-\{\pp+i\})$  is either equal to $\{\theta_p - x_p\}$ or $\{-x_p\}$, and the state $x_p$ converges exponentially to $\theta_p$ or $0$, respectively.
	%Similar to the proof of \Cref{lem_equilibriumPointofContinuous}, it is straightforward to show that
	Also, the final asymptotic value of $x_p$ is the same as the $(\pp + p)^\text{th}$ component of the equilibrium point $\bm q.$
	As for the imitative types, depending on the sign of $u^\1(q^*_k) - u^\2(q^*_k)$, for all $p \in [\pp]$ either $x_p$ converges to $\theta_p$ or $0$, which is the same as the $p^\text{th}$ element of vector $\bm q.
	$
	Recalling that for any $y, z \in \mathbb{R}$ $\vert y + z\vert \leq \vert y \vert + \vert z \vert $ and given $
	\vert x_{\pp + i}(t) - q_{\pp + i} \vert =\vert ( \sum_{p=1}^{2\pp}x_p(t) - q^*_k) -  \sum_{l \neq (\pp +i)} (x_l(t) - q_l) \vert,
	$
	we have 
	%\begin{equation} \label{proof_1}
	%\scalebox{0.9}{$
		$\vert x_{\pp + i}(t)\! -\! q_{\pp + i} \vert \leq \vert ( \sum_{p=1}^{2\pp}x_p(t) - q^*_k) \vert + \vert \sum_{l \neq (\pp + i)} (x_l(t) - q_l) \vert.$
		%$}
	%\end{equation}
	Convergence of both terms on the right-hand side of the preceding relation to zero implies $x_{\pp +i}(t) \to (q^*_k - \eta_{i-1} - \eta_{j}')$ if $u^\1(q^*_k) < u^\2(q^*_k)$
	and
	$x_{\pp + i}(t) \to (q^*_k - \eta_{i-1} - \eta_{j}' - \zeta_0)$ if $u^\1(q^*_k) > u^\2(q^*_k)$.
	This completes the proof for Item 1.
	Item 2)
	Here, $u^{\1}(q^*_k) = u^{\2}(q^*_k)$, and instead of an isolated equilibrium point, the dynamics \eqref{eq:type_mixed} admit a continuum of equilibria.
	% We consider the case where the point $\bm q$ is a mixed nonconformist equilibrium.
	We have $\tau_i = \bm 1^\top \bm q$ for some $i \in [\p]$.
	Similar to the proof of the first part, i.e., $(u^\1(q^*_k) \neq u^\2(q^*_k))$, it can be shown that
	for an $\epsilon$ satisfying  $\big((\tau_i-\epsilon, \tau_i+\epsilon) - \{\tau_i\}\big) \cap \big(\mathcal{C} \cup \mathcal{T} \big) = \varnothing,$ the population proportion of $\1$-players,
	which is governed by 
	\eqref{eq:populationProportion},
	reaches the $\epsilon$-neighborhood of the point $\tau_i$ in finite time and remains there.
	Accordingly, the element $p$ of the state $\x$, except for the first $\pp$ elements (which are associated with the imitative types) and element $\pp + i$ (whose threshold equals $q^*_k$), i.e., $p \in \{\pp+1,\ldots, 2\pp \} - \{\pp + i \},$ converges exponentially to $q_p$, which is the  $p^\text{th}$ component of the equilibrium point $\bm q.$
    \tb{The values of these $\pp - 1$ components of the point $\bm q$ add up to $\eta_{i-1} + \eta'_j$, where $j = \max \{k\leq \p' \mid \tau'_j < \tau_i\}$}.
	With this and the fact that $\sum_{i=1}^{2\pp} x_i$ converges to $q^*_k = \tau_i$, it can be concluded that both terms on the right-hand side of the inequality
	%  \begin{equation*} 
		%  \scalebox{0.9}{$
			% \begin{aligned}
				$\vert (x_{\pp+i}(t) + \sum_{p=1}^{\pp}x_p(t)) - (\tau_i - \eta_{i-1} - \eta'_j) \vert \leq 
				\vert (\sum_{p=1}^{2\pp}x_p(t) - \tau_i) \vert  + \vert \sum_{\substack{\pp<l\leq 2\pp\\ l \neq i+\pp}} (x_l(t) - q_l) \vert$
				%\end{aligned}
				%$}
			%\end{equation*}
			converge to zero, so does its left side.
			Item 3 can be handled similarly to the second item.
			The proof of the last item is similar to that of the second part of Theorem 2 in \cite{aghaeeyan2023discrete} and is omitted.
		\end{proofref}
		\begin{proofref}{\Cref{thm_birkhoff_center_invariant_population}}
			Given \Cref{lemGSAPS}, the mixed population dynamics Markov chain with transition probabilities \eqref{eq:markov}, for sequence $\langle \frac{1}{\N} \rangle_{\N \in \mathcal{N}}$, is a GSAP for the dynamics \eqref{eq:type_mixed}.
			The transition probability \eqref{eq:markov} is homogeneous, and the state space for each population size $\N \in \mathcal{N}$ is finite.
			Hence, the Markov chain admits an invariant measure.
			On the other hand, every solution of the dynamics \eqref{eq:type_mixed} converges to some members of the set $\bm{\mathcal{Q}}$ (\Cref{lem:attraction-population-single-equilibrium}).
			This implies that the set $\bm{\mathcal{Q}}$ is the set of recurrent points of \eqref{eq:type_mixed}. 
			Since the set $\bm{\mathcal{Q}}$ includes some isolated equilibria and/or continua of equilibrium points, it is closed, and hence
			the Birkhoff center of the dynamics \eqref{eq:type_mixed} equals the set $\bm {\mathcal{Q}}.$
			With this and given 
			\Cref{thm:implicationOFSandholm}, the support of the limit of the invariant measure is contained in any open set containing the set $\bm{\mathcal{Q}}.$
		\end{proofref}
\end{document}